\documentclass{article}

\usepackage[utf8x]{inputenc}
\usepackage{amsmath,ifthen,color,eurosym}

\usepackage{wrapfig}
\usepackage{hyperref}
\usepackage{latexsym,amsthm,amsmath,amssymb,url}
\usepackage{fancyhdr}
\usepackage{graphicx}
\usepackage{comment}
\usepackage{stmaryrd}
\usepackage{enumerate}
\usepackage{cite}
\usepackage{tikz}
\usepackage{framed}

\hypersetup{
  pdftitle = {An O(log OPT)-approximation for covering and packing minor
    models of the pumpkin},
  pdfauthor = {D. Chatzidimitriou, J-.F. Raymond, I. Sau and D. M. Thilikos},
  colorlinks = true,
  linkcolor = black!30!red,
  citecolor = black!30!green
}

\usetikzlibrary{shapes}
\usetikzlibrary{decorations.pathreplacing}
\tikzset{black node/.style={draw, circle, fill = black, minimum size = 7pt, inner sep = 0pt}}
\tikzset{diam/.style={draw, diamond, fill = black, minimum size = 7pt, inner sep = 0pt}}
\tikzset{white node/.style={draw, circle, fill = white, minimum size = 7pt, inner sep = 0pt}}
\tikzset{normal/.style = {draw=none, fill = none}}
\tikzset{bag/.style = {draw, circle, fill = white, minimum size = 0.75cm, inner sep = 0pt}}

\makeatletter
\tikzset{circle split part fill/.style  args={#1,#2}{%
 alias=tmp@name, %
  postaction={%
    insert path={
     \pgfextra{% 
     \pgfpointdiff{\pgfpointanchor{\pgf@node@name}{center}}%
                  {\pgfpointanchor{\pgf@node@name}{east}}%            
     \pgfmathsetmacro\insiderad{\pgf@x}
      \fill[#1] (\pgf@node@name.base) ([xshift=-\pgflinewidth]\pgf@node@name.east) arc
                          (0:180:\insiderad-\pgflinewidth)--cycle;
      \fill[#2] (\pgf@node@name.base) ([xshift=\pgflinewidth]\pgf@node@name.west)  arc
                           (180:360:\insiderad-\pgflinewidth)--cycle;            %  \end{scope}   
         }}}}}  
 \makeatother
\definecolor{colork4}{rgb}{0.8, 0.3, 0.2}
    \definecolor{colork23}{rgb}{0.15, 0.65, 0.85}
\tikzset{cnode/.style = {draw, circle, minimum size = 7pt, shape = circle split, circle split part fill={colork23, colork4}}}
\usepackage{todonotes}
\newtheorem{theorem}{Theorem}
\newtheorem{proposition}{Proposition}
\newtheorem{lemma}{Lemma}
\newtheorem{corollary}{Corollary}
\newtheorem{observation}{Observation}

\newcommand{\cupall}{\pmb{\pmb{\cup}}}
\newcommand{\diss}{\mbox{\sf diss}}

\newcommand{\ex}{{\sf ex}}
\newcommand{\pack}{\mbox{\sf pack}}
\newcommand{\children}{{\sf children}}
\newcommand{\cover}{\mbox{\sf cover}}
\DeclareMathOperator{\xcover}{{\sf x-cover}}
\DeclareMathOperator{\xpack}{{\sf x-pack}}
\DeclareMathOperator{\vcover}{{\sf v-cover}}
\DeclareMathOperator{\vpack}{{\sf v-pack}}
\DeclareMathOperator{\ecover}{{\sf e-cover}}
\DeclareMathOperator{\epack}{{\sf e-pack}}
\DeclareMathOperator{\sig}{{\sf sig}}

\newcommand{\N}{\mathbb{N}}

\DeclareMathOperator{\Pkg}{\bf {P}} % packings
\DeclareMathOperator{\Cvg}{\bf {C}} % coverings

\DeclareMathOperator{\tpw}{\bf {tpw}} % treepartitionwidth

\newcommand{\intv}[2]{\left \llbracket #1, #2 \right \rrbracket}
\newcommand{\x}{{\sf x}}
\newcommand{\Es}{E^{\sf s}}

\newcounter{func}
\newcommand{\newfun}[1]{f_{\refstepcounter{func}\label{#1}\thefunc}} % declare a new numbered function with label given by the argument
\newcommand{\funref}[1]{\hyperref[#1]{f_{\ref*{#1}}}} % print a
\newcommand{\eaddress}[1]{\href{mailto:#1}{\texttt{#1}}}

\title{An $O(\log \mathrm{OPT})$-approximation for covering and packing minor models of $\theta_r$\thanks{This research have been supported by
the Warsaw Centre of Mathematics and Computer Science (Jean-Florent Raymond), by the
    (Polish) National Science Centre grant PRELUDIUM 2013/11/N/ST6/02706 (Jean-Florent Raymond), and by the E.U. (European Social Fund - ESF) and Greek national funds through the Operational Program ``Education and Lifelong Learning'' of the National Strategic Reference Framework (NSRF) - Research Funding Program: ``Thales. Investing in knowledge society through the European Social Fund'' (Dimitris Chatzidimitriou and Dimitrios M. Thilikos).     
Email addresses:
    \eaddress{hatzisdimitris@gmail.com},
    \eaddress{jean-florent.raymond@mimuw.edu.pl},
    \eaddress{ignasi.sau@lirmm.fr}, and
    \eaddress{sedthilk@thilikos.info}.}}

\author{Dimitris Chatzidimitriou\thanks{Department of Mathematics, National and Kapodistrian  University of Athens, Greece.} \and Jean-Florent Raymond\thanks{Faculty of Mathematics, Informatics and Mechanics, University of Warsaw, Poland.}~\thanks{AlGCo project team, CNRS, LIRMM, France.} \and  Ignasi Sau$^\S$ \and Dimitrios M. Thilikos$^\dag$$^\S$}

\date{}

\begin{document}
\maketitle

\begin{abstract}\noindent 
%Let ${\cal C}_{H}$ be the class of graphs containing some fixed graph $H$ as a minor.
Given two graphs $G$ and $H$, we define $\vcover_{H}(G)$ (resp. $\ecover_{H}(G)$) as  
the minimum number of vertices (resp. edges) whose removal from $G$ produces a graph without 
any minor  isomorphic to ${H}$. 
Also   $\vpack_{H}(G)$ (resp. $\epack_{H}(G)$) is the maximum number of vertex- (resp. edge-) disjoint subgraphs of $G$ that contain a minor isomorphic to $H$.
We denote by $\theta_{r}$ the graph with two vertices and $r$ parallel edges between them.
When $H=\theta_{r}$, the parameters $\vcover_{H}$, $\ecover_{H}$, $\vpack_{H}$, and $\epack_{H}$ are {\sf NP}-hard to compute (for
sufficiently big values of~$r$). 
Drawing upon combinatorial results in~\cite{ChatzidimitriouRST15mino}, we give an algorithmic proof that if $\vpack_{{\theta_{r}}}(G)\leq k$, then $\vcover_{\theta_{r}}(G) = O(k\log k)$, and similarly for $\epack_{\theta_{r}}$ and~ $\ecover_{\theta_{r}}$.
In other words, the class of graphs containing ${\theta_{r}}$ 
as a minor has the vertex/edge
Erdős-Pósa property, for every positive integer $r$. Using the algorithmic machinery of our proofs we introduce a unified approach for the 
design of an  $O(\log {\rm OPT})$-approximation algorithm
for~$\vpack_{{\theta_{r}}}$, $\vcover_{{\theta_{r}}}$, 
$\epack_{{\theta_{r}}}$, and $\ecover_{{\theta_{r}}}$ that runs in
$O(n\cdot \log(n)\cdot m)$ steps.
Also, we derive several new Erdős-Pósa-type~results from the techniques that we introduce.
\smallskip

\smallskip

\end{abstract}
\section{Introduction}

All graphs in this paper are undirected, do not have loops but they may contain multiple edges. 
We denote by $\theta_{r}$ the graph containing two vertices $x$ and $y$
connected by $r$ parallel edges between $x$ and $y$.  Given a graph class ${\cal C}$ and a graph $G$,
we call {\em ${\cal C}$-subgraph} of $G$ any subgraph of $G$ that is isomorphic to some graph in ${\cal C}$.
All along this paper, when giving the running time of an algorithm
with input some graph $G$, we agree that  $n=|V(G)|$ and $m=|E(G)|$.

\paragraph{Packings and coverings.}
Paul Erdős and Lajos Pósa, proved in 1965\cite{ErdosP65} that there is a function $f\colon \N \rightarrow \N$ 
such that for each positive integer $k$, every graph either contains $k$ 
vertex-disjoint cycles or it contains $f(k)$ vertices that intersect every cycle in $G$. 
Moreover, they proved that the ``gap" of this min-max relation is  $f(k) = O(k\cdot  \log k)$ and that this gap is optimal.
This result initiated an interesting line of research on the duality between coverings 
and packings of combinatorial objects.  To formulate this duality, given a  class ${\cal C}$ of connected graphs,  
we define $\vcover_{\cal C}(G)$ (resp. $\ecover_{\cal C}(G)$) as the 
minimum cardinality of a set $S$ of vertices (resp. edges) such that 
each ${\cal C}$-subgraph of $G$ contains some element of $S$.
Also, 
we define  $\vpack_{\cal C}(G)$ (resp. $\epack_{\cal C}(G)$) as 
the maximum number of vertex- (resp. edge-) disjoint ${\cal C}$-subgraphs of $G$.

We say that ${\cal C}$ has the {\em vertex Erdős–Pósa property} (resp. the  {\em edge Erdős–Pó\-sa property})  if
there is a function $f:\N \rightarrow \N$, called {\em gap function},
such that, for 
every graph $G$,  $\vcover_{\cal C}(G)\leq f(\vpack_{\cal C}(G))$ (resp.  
$\ecover_{\cal C}(G)\leq f(\epack_{\cal C}(G))$). Using this terminology, 
the original result of Erdős and Pósa says that the set of all cycles 
has the  vertex Erdős–Pósa property with gap $O(k\cdot \log k)$. 
The general question in this branch of Graph Theory is to 
detect instantiations of ${\cal C}$ which have the 
vertex/edge Erdős–Pósa property (in short, {\em {\sf v}/{\sf e}-EP-property}) 
and when this is the case, minimize   
the gap function $f$. 
Several theorems of this type have been proved concerning 
different instantiations of ${\cal C}$ such as odd
cycles~\cite{KawarabayashiN2007erdo,RautenbachR01thee}, long
cycles~\cite{BirmeleBR07erdo}, and graphs containing cliques as
minors~\cite{DiestelKW12thee} (see
also~\cite{ReedRST96pack,KakimuraKK12erdo,GeelenK09thee} for results
on more general combinatorial structures and
\cite{Reed97tree,Raymond2016recent} for surveys).

A general class  that is known to have the {\sf v}-EP-property  is the
class ${\cal C}_{H}$ of the graphs  that contain some fixed planar
graph $H$ as a minor\footnote{A graph $H$ is a {\em minor} of a graph
  $G$ if it can be obtained from some subgraph of $G$ by contracting
  edges.}. This fact 
was proven by Robertson and Seymour in~\cite{RobertsonS86GMV} and 
the best known general gap is $f(k)=O(k\cdot \log^{O(1)}k)$ due to the results 
of~\cite{ChekuriC13larg}  --- see also~\cite{FominLMPS13,FioriniJW13excl,FioriniJS13} for better gaps
for particular instantiations of  $H$. Moreover, the planarity 
of $H$ appears to be the right dichotomy, as 
for non-planar $H$'s, ${\cal C}_{H}$ {\sl  does not}  
satisfy the  {\sf v}-EP-property.
Besides the near-optimality of the general upper bound of~\cite{ChekuriC13larg}, 
it is open whether 
the lower bound $\Omega(k\cdot \log k)$ can be matched for the general 
gap function, while this is indeed the case when
$H=\theta_{r}$~\cite{FioriniJS13}.

The question about classes that have  the  {\sf e}-EP-property has also attracted some attention (see~\cite{BirmeleBR07erdo}). According to~~\cite[Exercise 23 of Chapter 7]{Diestel05grap}, the original proof of Erdős and Pósa implies that 
cycles have the  {\sf e}-EP-property with gap $O(k\cdot\log k)$. Moreover, as proved in~\cite{2013arXiv1311.1108R}, 
the class ${\cal C}_{\theta_{r}}$ has the {\sf e}-EP-property 
with the (non-optimal) gap $f(k)=O(k^2\cdot \log^{O(1)} k)$. Interestingly, not much more is 
known on the graphs $H$ for which ${\cal C}_{H}$ 
has the {\sf e}-EP-property and is tempting to conjecture that the planarity of $H$ 
provides again the right dichotomy. 
Other graph classes that are known to have the 
  {\sf e}-EP-property are rooted cycles~\cite{PontecorviW2012disj}  (here the cycles to be covered and packed are required to intersect some particular set of terminals of $G$)
  and odd cycles for the case where $G$ is a $4$-edge connected graph~\cite{KawarabayashiK2012edge}, a planar graph~\cite{KralV2004edge}, or a graph embeddable in an orientable surface~\cite{KawarabayashiN2007erdo}.

\paragraph{Approximation algorithms.}
The above defined four graph parameters are already quite general when ${\cal C}:={\cal C}_{H}$. 
%For simplicity, form now on, we 
%use the notations  $\vcover_{H}$, $\ecover_{H}$, $\vpack_{H}$, and $\epack_{H}$ instead of $\vcover_{{\cal C}_{H}}$, $\ecover_{{\cal C}_H}$, $\vpack_{{\cal C}_H}$, and $\epack_{{\cal C}_H}$.
From the algorithmic point of view, the computation of  
$\xpack_{{\cal C}_H}$ (for $\x \in \{ {\sf v}, {\sf e}\}$) corresponds  to the general family of graph 
packing problems, while  
the computation of $\xcover_{{\cal C}_H}$ belongs to the 
general family  of graph modification problems where the modification operation is the removal of vertices/edges (depending on whether $\x = {\sf v}$ or $\x = {\sf e}$). Interestingly, particular instantiations of $H=\theta_{r}$ 
generate known, well studied,  {\sf NP}-hard problems. For instance, asking whether 
$\vcover_{{\cal C}_{\theta_{r}}}\leq k$ generates  {\sc Vertex Cover} for $r=1$, {\sc Feedback Vertex Set} for 
$r=2$, and {\sc  Diamond Hitting Set} for $r=3$~\cite{FioriniJP10inte,FominMNL12}. Moreover, asking 
whether $\xpack_{{\cal C}_{\theta_{2}}}(G)\geq k$ corresponds to {\sc
  Vertex Cycle Packing}~\cite{BodlaenderTY11kern, Klo02} and {\sc Edge Cycle 
Packing}~\cite{Caprara2003239, Krivelevich:2007:AAH:1290672.1290685} when $\x={\sf v}$ and ${\sf x}={\sf e}$, respectively. Finally,  
asking whether 
$|E(G)|-\ecover_{{\cal C}_{\theta_{3}}}(G)\leq k$  corresponds to the {\sc Maximum Cactus Subgraph}\footnote{The  {\sc Maximum Cactus Subgraph} problem asks, given a graph $G$ and an  integer $k$, whether $G$ contains a subgraph with $k$ edges where no two cycles share an edge. 
It is easy to reduce to this problem the {\sc Vertex Cycle Packing} problem 
on cubic graphs which, in turn, can be proved to be {\sf NP}-complete using a simple variant of the {\sf NP}-completeness proof of  
{\sc Exact Cover by 2-Sets}~\cite{Golovach15}.}. All parameters keep being {\sf NP}-complete to compute
because the aforementioned base cases can be reduced to the general one by replacing each edge by one 
of multiplicity $r-1$.\medskip

From the approximation point of view, 
it was proven 
in~\cite{FominMNL12} that, when $H$ is a planar graph, there is a  
randomized polynomial $O(1)$-approximation algorithm for  $\vcover_{{H}}$.
For the cases of $\vcover_{{\cal C}_{\theta_{r}}}$ and $\vpack_{{\cal C}_{\theta_{r}}}$,
$O(\log n)$-approximations are
known for every $r\geq 1$ because of~\cite{JoretPSST14hitt} (see also~\cite{SalavatipourV05disj}).
Moreover,  $\vcover_{{\cal C}_{\theta_{r}}}$ admits a deterministic $9$-approximation~\cite{FioriniJP10inte}.
Also, about $\epack_{{\cal C}_{\theta_{r}}}(G)$
it is known, from~\cite{KrivelevichNY05appr}, that  
there is a polynomial
$O(\sqrt{\log n})$-approximation algorithm for  the case where $r=2$. 
Notice also that for $r=1$, it is trivial to compute  $\ecover_{{\cal C}_{\theta_{r}}}(G)$ in polynomial time. 
However, to our knowledge, nothing  is known about the computation of $\ecover_{{\cal C}_{\theta_{r}}}(G)$ for $r\geq 3$.

\paragraph{Our results.}
In this paper we introduce a unified approach for the study of the combinatorial interconnections and the approximability 
of the parameters $\vcover_{{\cal C}_{\theta_{r}}}$, $\ecover_{{\cal C}_{\theta_{r}}}$,   $\vpack_{{\cal C}_{\theta_{r}}}$, and $\epack_{{\cal C}_{\theta_{r}}}$.
Our main combinatorial result is the following.
\begin{theorem}
\label{madin}
For every $r\in\N_{\geq 2}$ and every ${\sf x}\in\{{\sf v},{\sf e}\}$  the graph class 
${\cal C}_{\theta_{r}}$ has the {\em {\sf x}-EP-property} with (optimal) gap function $f(k)=O(k\cdot \log k)$.
\end{theorem}

Our proof is unified and treats simultaneously the covering and the packing parameters for both the 
vertex and the edge cases.
This  verifies the optimal combinatorial bound for the case where ${\sf x}={\sf v}$ \cite{FioriniJS13} and 
optimally improves the previous bound in~\cite{2013arXiv1311.1108R} for the case where  ${\sf x}={\sf e}$.
Based on the proof of \autoref{madin}, we prove the following algorithmic~result.
\begin{theorem}
\label{admndi}
For every $r\in \N_{\geq 2}$ and every ${\sf x}\in\{{\sf v},{\sf e}\}$,
there is an algorithm that, given a graph $G$, outputs an
$O(\log {\rm OPT})$-approximation 
of~$\xcover_{{\cal C}_{\theta_{r}}}(G)$ and~$\xpack_{{\cal
    C}_{\theta_{r}}}(G)$ in $O(n\cdot \log(n)\cdot m)$ steps when $\x =
{\sf v}$ and in $O(m^2 \log n)$ steps when $\x = {\sf e}$.
\end{theorem}

\autoref{admndi}
improves the results in~\cite{JoretPSST14hitt} for the cases of
$\vcover_{{\cal C}_{\theta_{r}}}$ and $\vpack_{{\cal C}_{\theta_{r}}}$
and, to our knowledge, is the first approximation algorithm for
$\ecover_{{\cal C}_{\theta_{r}}}$ and $\epack_{{\cal C}_{\theta_{r}}}$ for $r\geq 3$.
We were also able to derive the following Erdős-Pósa-type result with
linear gap on graphs of bounded tree-partition width (the definition of this 
width parameter is given in~\autoref{boundgraphs}).
\begin{theorem}\label{t:epbound}
 Let $t \in \N$. For every ${\sf x}\in\{{\sf v},{\sf e}\}$, the
  following holds: if ${\cal H}$ is a finite collection of connected graphs and $G$ is
  a graph of tree-partition width at most $t$, then ${\sf x}$-\cover$_{\cal
    H}(G)\leq \alpha \cdot {\sf x}$-\pack$_{\cal H}(G)$, where $\alpha$ is a constant which depends only on $t$ and $\mathcal{H}$.
\end{theorem}

Let $\theta_{r,r'}$ (for some $r,r'\in \N_{\geq 1}$) denote the graph obtained
by taking the disjoint union of $\theta_{r}$ and $\theta_{r'}$ and
identifying one vertex of $\theta_{r}$ with one of $\theta_{r'}$.
Another  consequence of our results is that, for every $r,r' \in \N_{\geq 1}$, the
class ${\cal C}_{\theta_{r,r'}}$ has the edge Erdős–Pósa property. 
\begin{theorem}\label{t:epdoublepk}
For every $r,r' \in \N$, there is a function $\newfun{f:ep}^{r,r'}
\colon \N_{\geq 1}\rightarrow\N_{\geq 1}$ such that for every simple graph
$G$ where $k=\epack_{{\cal C}_{\theta_{r,r'}}}(G)$, it holds that
$\ecover_{{\cal C}_{\theta_{r,r'}}}(G)\leq \funref{f:ep}^{r,r'}(k)$.
\end{theorem}

\paragraph{Our techniques.}
Our proofs are based on the notion of {\sl partitioned protrusion} that,
roughly, is a tree-structured subgraph of $G$ with small boundary to the rest
of $G$ (see \autoref{boundgraphs} for the precise definition).
Partitioned protrusions were essentially introduced in~\cite{ChatzidimitriouRST15mino} by the name {\sl edge-protrusions}
and can be seen as the edge-analogue of the notion
of protrusions introduced in~\cite{BodlaenderFLPST09} (see also~\cite{BodlaenderFLPST09arxiv}). Our approach makes 
strong use of the main result of~\cite{ChatzidimitriouRST15mino}, that is 
equivalently stated as \autoref{rk6ter} in this paper.
According to this result,
there exists a polynomial 
algorithm that, given a graph $G$ and an integer $k$ as an input,
outputs one of the following: (1) a collection of $k$ edge/vertex 
disjoint ${\cal C}_{\theta_{r}}$-subgraphs of $G$, (2) 
 a ${\cal C}_{\theta_{r}}$-subgraph $J$ with $O(\log k)$ edges, or (3) a large {partitioned protrusion} of $G$.

Our approximation algorithm does the following for each $k\leq |V(G)|$.
If the first case of the above combinatorial result applies on $G$, we can safely 
output a  packing of $k$  ${\cal C}_{\theta_{r}}$-subgraphs in $G$.
In the second case, we make some progress as we may remove the vertices/edges of $J$ 
from $G$ and then set $k:=k-1$. 
In order to deal with the third case, 
we prove that in a graph $G$ with a sufficiently large partitioned protrusion, we can either find  some  ${\cal C}_{\theta_{r}}$-subgraph with $O(\log k)$ edges (which is the same as in the second case), or we can  replace it by a smaller  graph where both 
$\xcover_{\cal H}$ and $\xpack_{\cal H}$ remain invariant (\autoref{l:new_reduce}).
The proof that such a reduction is possible is 
given in \autoref{sec:red} and is based 
on a suitable dynamic programming encoding of partial packings 
and coverings that is designed to work on partitioned protrusions.

Notice that  the ``essential step'' in the above procedure is the second case that reduces the packing number of the current graph 
by 1 to the price of reducing  the covering number by $O(\log k)$. This is the main 
argument that supports the claimed $O(\log {\rm OPT})$-approximation algorithm (\autoref{admndi}) and the corresponding Erdős–Pósa relations in~\autoref{madin}. Finally, \autoref{t:epbound} is a combinatorial implication of \autoref{l:new_reduce} and \autoref{t:epdoublepk}
follows by combining \autoref{t:epbound} with the results of  Ding and  Oporowski in~\cite{Ding199645}.

\paragraph{Organization of the paper.} In \autoref{prels} we provide all concepts
and notation that we use in our proofs. 
%\autoref{l:new_reduce}, that is the main technical part of the paper,  is proven in \autoref{sec:red}.
\autoref{sec:red} contains the proof of \autoref{l:new_reduce}, which is the main technical part of the paper.
The presentation and analysis of our approximation algorithm is done in \autoref{hki6yqpo9},
where \autoref{madin} and~\autoref{admndi} are proven. \autoref{i9o0i4y} contains the proofs of \autoref{t:epbound} and~\autoref{t:epdoublepk}. Finally, we summarize our results and provide several directions for further research in~\autoref{sec:concl}.

\section{Preliminaries}
\label{prels}

\subsection{Basic definitions}
\label{basicdefs}

%Given a function $\phi: A\rightarrow B$ and a set $C\subseteq A$, we define the {\em restriction} of $\phi$
%to $C$ as $\phi|_{C}=\{(x,\phi(x))\mid x\in C\}$.
%Also, we define $\phi(C)=\{\phi(x)\mid x\in C\}$.

Let ${\bf t}=(x_{1},\ldots,x_{l})\in \N^{l}$ and $\chi,\psi: \N
\rightarrow \N$.
We say that $\chi(n)=O_{{\bf t}}(\psi(n))$ if there exists a computable
function $\phi:\N^{l} \rightarrow \N$
such that  $\chi(n)=O( \phi({\bf t})\cdot \psi(n))$.
For every $i,j\in \N$, we denote by $\intv{i}{j}$ the interval of $\{i,
\dots, j\}$.

\paragraph{Graphs.} All graphs in this paper are undirected, loopless, and may have multiple edges. 
For this reason, a graph $G$ is represented by a pair $(V,E)$ where $V$
is its vertex set, denoted by $V(G)$ and $E$ is its edge multi-set,
denoted by $E(G)$. The {\em edge-multiplicity} of an edge of $G$ is
the number of times it appears in~$E(G)$. We set $n(G)=|V(G)|$ and
$m(G)=|E(G)|$. If ${\cal H}$ is a finite collection of connected graphs, we set $n({\cal
  H})=\sum_{H\in {\cal H}} n(H)$, $m({\cal H})=\sum_{H\in {\cal H}}
m(H)$, and $\cupall {\cal H} = \bigcup_{G \in {\cal H}}G$. 
We denote by $\delta(G)$ the \emph{minimum degree} of $G$, that is the
minimum number of neighbors a vertex can have in~$G$.

Let ${\sf x}\in \{{\sf v}, {\sf e}\}$ where, in the rest of this paper,  {\sf v} will be
interpreted as ``vertex'' and {\sf e} will be interpreted as ``edge''.
Given a graph $G$, we
denote by ${A}_{\sf x}(G)$ the set of vertices or edges of $G$
depending on whether $\x={\sf v}$ or $\x ={\sf e}$, respectively.

Given a graph $H$ and a graph $J$ that are both subgraphs of the same graph $G$,
we define the subgraph  $H\cap_{G} J$ of $G$ as the graph $(V(H)\cap V(J),E(H)\cap E(J))$.

Given a graph $G$ and a set $S\subseteq V(G)$, such that all vertices
of $S$  have degree 2 in $G$,
we define $\diss(G,S)$ as the graph obtained 
from $G$ after we dissolve in it all vertices in $S$, i.e.,
replace each maximal path whose internal vertices are in $S$ with an edge whose endpoints are 
the endpoints of the path.

A \emph{rooted tree} is a pair $(T,s)$ where $T$ is a tree and $s \in V(T)$ is a distinguished vertex called the \emph{root} of $T$. If $v\in V(T)$, a vertex $u \in V(T)$ is a \emph{descendant} of $v$ in $(T,s)$ if $v$ lies on the (unique) path from $u$ to $s$. Note that~$v$ is a descendant of itself. We denote by $\children_{(T,s)}(v)$ the set of children of~$v$, which are the vertices that are both neighbors and descendants of~$v$.
% We denote by $\desc_{(T,s)}(v)$ the set of descendants of $v$ in~$(T,s)$.

Given a graph $G$ and an edge $e=(u,v)\in E(G)$, the {\em subdivision of $e$} is defined as the operation of removing $e$ and adding a new vertex $w$ and the edges $(w,u)$ and $(w,v)$.
A graph $H$ is called a {\em subdivision} of $G$, if it can be obtained by performing a series of subdivisions in $G$.

\paragraph{Minors and topological minors.}
Given two graphs $G$ and $H$, we say that $H$ is a {\em minor} of $G$ if there exits 
some function $\phi: V(H)\rightarrow 2^{V(G)}$ such that 
\begin{itemize}
\item for every $v\in V(H)$, $G[\phi(v)]$ is connected;
\item for every two distinct $v,u\in V(H)$, $\phi(v)\cap \phi(u)=\emptyset$; and 
\item for every edge $e=\{v,u\}\in E(H)$ of multiplicity $l$, there are at least $l$ edges in $G$ with one endpoint in $\phi(v)$ 
and the other in $\phi(u)$.
\end{itemize}
We say that $H$ is a {\em topological minor} of $G$ if there exits 
some collection ${\cal P}$ of paths in $G$ and an
injection $\phi: V(H)\rightarrow V(G)$ such that 
\begin{itemize}
\item no path in ${\cal P}$ has an internal vertex that belongs to some other path in ${\cal P}$;
\item $\phi(V(H))$ is the set of endpoints of the paths in ${\cal P}$; and 
\item for every two distinct $v,u\in V(H)$, $\{v,u\}$ is an edge of $H$ of multiplicity $l$ if and only if there are $l$
paths in ${\cal P}$ between $\phi(v)$ and $\phi(u)$.
\end{itemize}

Given a graph $H$, we define by $\ex(H)$ the set of all
topologically-minor minimal graphs that contain $H$ as a minor.
Notice that the size of $\ex(H)$ is upper-bounded by some function of~$m(H)$ and
 that $H$ is a minor of $G$ if it contains a 
member  of $\ex(H)$ as  a topological minor.
An {\em $H$-minor model of $G$} is  any 
minimal subgraph of $G$ that contains a 
member  of $\ex(H)$ as  a topological minor.

\paragraph{Subdivisions.}
The subdivision of an edge $\{u,v\}$ of a graph is the operation that
introduces a new vertex $w$ to the graph adjacent to $u$ and $v$ and
deletes the edge. A graph $G$ is a subdivision of $H$ if it can be
obtained from $H$ by repeteadly subdividing edges.
If $G$ is a graph and  ${\cal H}$ is a finite collection of connected graphs,
an {\em ${\cal H}$-subdivision} of $G$ is a subgraph $M$ of $G$ that is a
subdivision of a graph, denoted by $\hat{M}$, that is isomorphic to a
member of ${\cal H}$.  Clearly, the vertices of $\hat{M}$ are vertices
of $G$ and its edges correspond to paths in $G$ between their
endpoints such that internal vertices of a path do not appear in any other path.
We refer to the vertices of $\hat{M}$ in $G$ as the \emph{branch vertices} of
the ${\cal H}$-subdivision $M$, whereas internal vertices of the paths between branch
vertices will be called~\emph{subdivision vertices} of ${M}$.
Note that we are here dealing with what is sometimes called a
topological minor model.
A graph which contains no $\mathcal{H}$-subdivision is said to be $\mathcal{H}$-free.
Notice that $G$ has an ${\cal H}$-subdivision iff $G$ contains a graph of
$\mathcal{H}$ as a topological minor.

\paragraph{Packings and coverings.}
An ${\sf x}$-{\em ${\cal H}$-packing of size $k$} of a graph $G$ is a collection
${\cal P}$ of $k$ pairwise ${\sf x}$-disjoint ${\cal H}$-subdivisions of $G$.
Given an ${\sf x}\in \{{\sf v}, {\sf e}\}$, we define $\Pkg_{{\sf
x},{\cal H}}^{\geq k}(G)$ as the set of all {\sf x}-${\cal
H}$-packings of $G$ of size at least $k$.    An $\x$-{\em ${\cal H}$-covering} of a graph $G$ is a set $C\subseteq
{A}_{\sf x}(G)$ such that $G\setminus C$ does not contain any ${\cal
H}$-subdivision. We define $\Cvg_{{\sf x},{\cal H}}^{\leq k}(G)$ as the set
of all {\sf x}-${\cal H}$-coverings of $G$ of size at most $k$.  We finally
define
\[ \xcover_{\cal H}(G)=\min\{k\mid \Cvg_{\x,{\cal H}}^{\leq
k}(G)\neq \emptyset\}
\]
and
\[ {\sf x}\mbox{-\pack}_{\cal H}(G)=\max\{k\mid \Pkg_{\x,{\cal
H}}^{\geq k}(G)\neq \emptyset\}.
\]
These numbers are well-defined: $G$ always has a packing
of size $0$ (the empty packing) and a covering of size $|{A}_{\sf x}(G)|$ (which contains
all vertices/edges of~$G$).
It is easy to observe that, for every  graph $G$ and every  finite collection of  connected graphs
${\cal H}$, the following inequalities hold: 
\begin{eqnarray*}
{\sf
v}\mbox{-\cover}_{\cal H}(G)\leq {\sf e}\mbox{-\cover}_{\cal H}(G), &&
{\sf v}\mbox{-\pack}_{\cal H}(G)\leq {\sf e}\mbox{-\pack}_{\cal
H}(G),\\
{\sf
v}\mbox{-\pack}_{\cal H}(G)\leq {\sf v}\mbox{-\cover}_{\cal H}(G), && 
{\sf e}\mbox{-\pack}_{\cal H}(G)\leq {\sf e}\mbox{-\cover}_{\cal H}(G).
\end{eqnarray*}
Notice that the definition of these invariants differs from that given in the
introduction.
 
\subsection{Boundaried graphs}
\label{boundgraphs}

Informally, a boundaried graph will be used to represent a graph that
has been obtained by ``dissecting'' a larger graph along some of 
its edges, where the boundary vertices correspond to edges that 
have been cut.

\paragraph{Boundaried graphs.}
A {\em  boundaried graph} ${\bf G}=(G,B,\lambda)$ is a triple consisting of a graph $G$,
where $B$ is a set of vertices of degree one (called \emph{boundary})
and ${\lambda}$ is a bijection from $B$ to a subset of
$\N_{\geq 1}$. The edges with at least one endpoint in $B$ are called \emph{boundary edges}.
We define $\Es(G)$ as the subset of $E(G)$ of boundary edges.
We stress that instead of $\N_{\geq 1}$ we could choose any other set of symbols to label 
the vertices of~$B$.
We denote the set of labels of ${\bf G}$ by $\Lambda({\bf
  G})=\lambda(B)$. Given a finite collection of connected 
  graphs ${\cal H}$, we say that a boundaried graph ${\bf G}$ is ${\cal H}$-free if $G\setminus B$ is  ${\cal H}$-free.

Two  boundaried graphs ${\bf G}_1$ and ${\bf G}_2$  are {\em
  compatible} if $\Lambda({\bf G}_{1})=\Lambda({\bf G}_{2})$. Let now
${\bf G}_1=(G_{1},B_{1},\lambda_{1})$ and ${\bf
  G}_2=(G_{2},B_{2},\lambda_{2})$ be two compatible boundaried
graphs. We define the graph ${\bf G}_1 {\oplus} {\bf G}_2$  as the
non-boundaried graph obtained by first taking the disjoint union of $G_{1}$ and $G_{2}$,
then, for every $i\in\Lambda({\bf G}_{1}),$ identifying $\lambda_{1}^{-1}(i)$ with  $\lambda_{2}^{-1}(i)$,
and finally dissolving all resulting identified vertices. Suppose that $e$ is an edge of $G={\bf G}_1 {\oplus} {\bf G}_2$
that was created after dissolving the vertex resulting from the identification 
of a vertex $v_{1}$ in $B_1$ and a vertex $v_{2}$ in $B_2$ and 
that  $e_{i}$ is the boundary edge of $G_{i}$ that has $v_{i}$ as endpoint, for $i=1,2$.
% Then we say that $e$ is the {\em heir} of $e_{i}$ in $G$, for $i=1,2$, and we denote this by $\heir_{G}(e_{i})$.
% For $i\in\{1,2\}$, if $S\subseteq E(G_{i})$, then 
% \[
% \heir_{G}(S)=(E(G_{i})\cap S)\cup \{\heir_{G}(e)\mid e\in E^{\sf
%   s}(G_{i})\cap S\}.
% \]
% For reasons of notational consistency, if $V\subseteq V(G_{i})$, we denote $\heir_{G}(S)=S$.

\autoref{fig:oplus} shows the result of the operation $\oplus$ on
two graphs. Boundaries are drawn in gray and their labels are
written next to them. The graphs ${\bf G}_1$ and ${\bf G}_2$ on this picture are
compatible as $\Lambda({\bf G}_1) = \Lambda({\bf G}_2) = \{0,1,2,3\}$.

\begin{figure}[ht]
  \centering
  \begin{tikzpicture}[every node/.style = black node, scale = 0.5, font=\small]
    \draw[rounded corners] (3,3.5) -- (3, -3.5) -- (0,-3.5) .. controls ++(-3,0)  and ++(-3,0) .. (0, 3.5) -- cycle;
    %%left graph
    \draw (2, 3) node[fill=gray!120!black, label = 0:0] {}
    -- (1, 2) node (g3) {}
    -- (-1, 2) node (g1) {}
    -- (-1, -2) node {}
    -- (1, -2) node (g4) {}
    -- (2, -3) node[fill = gray!120!black, label=0:3] {}
    (2,1) node[fill = gray!120!black, label = 0:1] {}
    -- (g3)
    (2, -1) node[fill = gray!120!black, label = 0:2] {} -- (g4) -- (g1);
    \begin{scope}[xshift = 6cm]%% right graph
      \draw[rounded corners] (-1, 3.5) -- (-1, -3.5) -- (0, -3.5) .. controls ++(3,0)  and ++(3,0) .. (0, 3.5) -- cycle;
      \draw (0, 3) node[fill= gray!120!black, label=180:2] {} to[bend left] (0,1) node[fill = gray!120!black, label = 180:1] {}
      (0,-1) node[fill = gray!120!black, label = 180:0] {} -- (1, -2) node {} -- (0,-3) node[fill = gray!120!black, label = 180:3] {};
    \end{scope}
    \begin{scope}[xshift = 13cm]
      \draw (0.75, 2.75) .. controls ++(-3,0)  and ++(-3,0) .. (0.75,
      -2.75) -- (1.5, -2.75) .. controls ++(3,0)  and ++(3,0)
      .. (1.5,2.75) -- cycle;
      \draw (0 ,2) node (p1) {}
      -- ++(0, -4) node {}
      -- ++(2,0) node (p2) {}
      -- (p1)
      -- ++(2,0) node (p3) {}
      -- (p2)
      -- ++(1,2) node {}
      -- (p3);
    \end{scope}
    \node[normal] at (4, 0) {\scalebox{2}{$\oplus$}};
    \node[normal] at (10, 0) {\scalebox{2}{$=$}};
    \node[normal] at (1, -4.5) {\scalebox{1}{${\bf G}_1$}};
    \node[normal] at (7, -4.5) {\scalebox{1}{${\bf G}_2$}};
    \node[normal] at (14, -4.5) {\scalebox{1}{$G$}};
    
  \end{tikzpicture}
  \caption{Gluing graphs together: $G ={\bf G}_1 \oplus {\bf G}_2$.}
  \label{fig:oplus}
\end{figure}
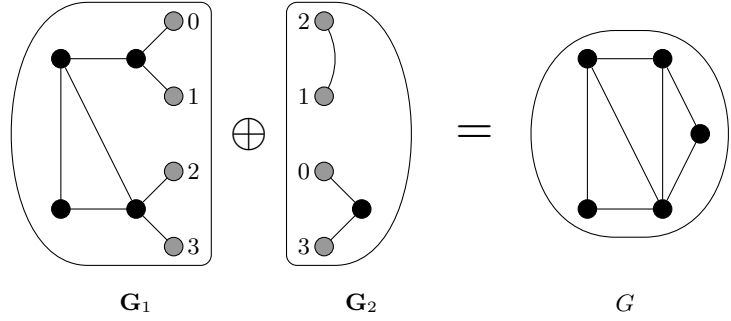

For every $t \in \N_{\geq 1}$, we denote by ${\cal B}_{t}$ all
boundaried graphs whose boundary is labeled by numbers in~$\intv{1}{t}$.
Given a  boundaried graph ${\bf G}=(G,B,\lambda)$ 
and a subset $S$ of $V(G)$ such that all vertices in $S$  have degree
2 in $G$, we define $\diss({\bf G},S)$ as the graph  
$\hat{\bf G}=(\hat{G},B,\lambda)$ where $\hat{G}=\diss(G,S)$.

Let $W$ be a graph and $S$ be a non-empty subset of $V(W)$.
An {\em
$S$-splitting} of $W$ is a pair $({\bf G}_{S},{\bf G}_{S^{\sf c}})$
consisting of two boundaried graphs ${\bf
G}_{S}=(G_{S},B_{S},\lambda_{S})$ and ${\bf G}_{S^{\sf c}}=(G_{S^{\sf
c}},B_{S^{\sf c}},\lambda_{S^{\sf c}})$ that can be obtained as
follows: First, let $W^+$ be the graph obtained by subdividing in $W$
every edge with one endpoint in $S$ and the other in $V(W)
\setminus S$ and let $B$ be the set of created vertices. Let $\lambda$
be any bijection from $B$ to a subset of $\N_{\geq 1}$. Then $G_S =
W^+[S \cup B]$, $G_{S^{\sf c}} = W^+ \setminus S$, $B_S = B_{S^{\sf
c}} = B$, and $\lambda_S = \lambda_{S^{\sf c}} = \lambda$.
Notice that there are infinite such pairs, depending on the numbers that will be assigned to the boundaries 
of ${\bf G}_{S}$ and ${\bf G}_{S^{\sf c}}$.
Moreover, keep in mind that all the boundary edges of $G_{S}$ 
are non-loop edges with exactly one endpoint in $B$ and the same holds for the boundary edges of 
$G_{S^{\sf c}}$. An example of a splitting is given in \autoref{fig:split}, where boundaries are depicted by gray vertices.
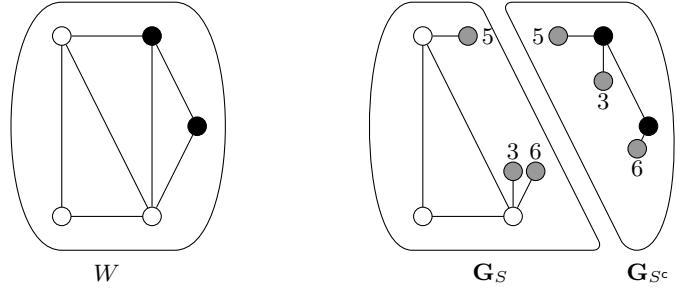
\begin{figure}[ht]
  \centering
  \begin{tikzpicture}[every node/.style = black node, font=\small,
    scale = 0.6]
    \begin{scope}[xshift = -2cm]
      \draw (0, 2.75) .. controls ++(-1.5,0)  and ++(-1.5,0) .. (0,
      -2.75) -- (2.5, -2.75) .. controls ++(1.5,0)  and ++(1.5,0)
      .. (2.5,2.75) -- cycle;
      \draw (0 ,2) node[white node] (p1) {}
      -- ++(0, -4) node[white node] {}
      -- ++(2,0) node[white node] (p2) {}
      -- (p1)
      -- ++(2,0) node (p3) {}
      -- (p2)
      -- ++(1,2) node {}
      -- (p3);
      \node[normal] at (1,-3.25) {$W$};
    \end{scope}
    \begin{scope}[xshift = 6cm]
      \draw[rounded corners] (0, 2.75) .. controls ++(-1.5,0)  and ++(-1.5,0) .. (0,
      -2.75) -- (4, -2.75) -- (1.25,2.75)  -- cycle;
      \draw (0 ,2) node[white node] (p1) {}
      -- ++(0, -4) node[white node] {}
      -- ++(2,0) node[white node] (p2) {}
      -- (p1)
      -- ++(1,0) node[fill = gray!120!black, label = 0:5] (p3h) {}
      (p2) -- ++(0,1) node[fill = gray!120!black, label=90:3] (p3h2) {}
      (p2) -- ++(0.5,1) node[fill = gray!120!black, label=6] (p5h) {};
      \node[normal] at (1.5,-3.25) {${\bf G}_S$};
    \end{scope}
    \begin{scope}[xshift = 9cm, rounded corners]
      \draw[xshift = -1.25cm] (2.75, -2.75) .. controls ++(1.5,0)  and ++(1.5,0)
      .. (2.5,2.75) -- (0, 2.75) -- cycle;
      \draw (2,0) node (p5){}
      -- ++(-1,2) node (p3){}
      -- ++ (-1,0) node[fill = gray!120!black, label = 180:5] {}
      (p3) -- ++(0,-1) node[fill = gray!120!black, label = -90:3] {}
      (p5) -- ++(-0.25,-0.5) node[fill = gray!120!black, label =-90:6] {};
      \node[normal] at (2,-3.25) {${\bf G}_{S^{\sf c}}$};
    \end{scope}
  \end{tikzpicture}
  \caption{Cutting a graph: $({\bf G}_{S}, {\bf G}_{S^{\sf c}})$ is an $S$-splitting of $W$, where $S$ consists of all the white vertices.}
  \label{fig:split}
\end{figure}

We say that ${\bf G}'=(G',B',\lambda')$ is a \emph{boundaried subgraph}
of ${\bf G}=(G,B,\lambda)$ if $G'$ is a subgraph of $G$, $B' \subseteq
B \cap V(G')$ and~$\lambda' = \lambda_{|B'}$.
On the other hand, ${\bf G}$ is a subgraph of a (non-boundaried) graph $H$ if ${\bf G} = {\bf H}_S$ for some $S$-splitting $({\bf H}_S, {\bf H}_{S^{\sf c}})$, where $S\subseteq V(H)$.

If $H$ is a graph, $G$ is a subgraph of $H$, and ${\bf F} =
(F,B,\lambda)$ is a boundaried subgraph of $H$, we define $G \cap_H
{\bf F}$ as follows. Let $S = V(G) \cap (V(F) \setminus B)$ and let
$G^+$ be the graph obtained by subdividing once every edge of $G$ that
has one endpoint in $S$ and the other in~$V(G) \setminus S$. We call
$B'$ the set of created vertices and let $G' = G^+[S \cup B']$. Then
$G'$ is a subgraph of $F$ where $B' \subseteq B$. For every $v \in
B'$, we set $\lambda'(v) = \lambda(v)$, which is allowed according to
the previous remark. Then $G \cap_H {\bf F} = (G', B',
\lambda')$. Observe that $G \cap_H {\bf F}$ is part of an
$S$-splitting of~$G$ (but not necessarily of~$H$).

Given two boundaried graphs ${\bf G}'=(G',B',\lambda')$ and ${\bf G}=(G,B,\lambda)$,  
 we say that they are {\em isomorphic} if there is an isomorphism 
from $G'$ to $G$ that respects the labelings of $B$ and $B'$, i.e., maps  every vertex $x\in B'$ to $\lambda^{-1}(\lambda'(x))\in B$.
Given a boundaried graph ${\bf G}=(G,B,\lambda)$, we denote
 $n({\bf G})=n(G)-|B|$ and $m({\bf G})=m(G)$.

Given a boundaried graph ${\bf G}=(G,B,\lambda)$ and an ${\sf
x}\in\{{\sf v},{\sf e}\}$, we set $A_{\sf x}({\bf G})=V(G)\setminus B$ or $A_{\sf
x}({\bf G})=E(G)$, depending on whether ${\sf x}={\sf v}$
or ${\sf x}={\sf e}$.

\paragraph{Partial structures.}
Let $\cal H$ be a family of graphs.
A boundaried subgraph $\bf J$ of a boundaried graph $\bf G$ is a
\emph{partial $\mathcal{H}$-subdivision} if there is a boundaried graph $\bf
H$ which is compatible with $\bf G$ and a boundaried subgraph ${\bf
  J}'$ of $\bf H$ which is compatible with $\bf J$ such that ${\bf J} \oplus {\bf J}'$ is an
$\mathcal{H}$-subdivision of ${\bf G} \oplus {\bf H}$.
Intuitively, this means that $\bf J$ can be extended into an $\cal
H$-subdivision in some larger graph.  In this case, the
$\mathcal{H}$-subdivision ${\bf J} \oplus {\bf J}'$ is said to be an \emph{extension} of $\bf J$.

Similarly, for every $p \in \N_{\geq 1}$, a collection of boundaried subgraphs
$\mathcal{J} = \{{\bf J}_1, \dots, {\bf
  J}_p\}$ of a graph $\bf G$ is a \emph{partial
  $\x$-$\mathcal{H}$-packing} if there is a boundaried graph $\bf H$ which is
compatible with $\bf G$ and a collection of boundaried subgraphs
$\{{\bf J}_1', \dots, {\bf
  J}_p'\}$ of $\bf H$ such that $\{{\bf J}_1 \oplus {\bf J}_1', \dots,
{\bf J}_p \oplus{\bf
  J}_p'\}$ is an $\x$-$\mathcal{H}$-packing of ${\bf G} \oplus {\bf
  H}$. The obtained packing  is said to be an \emph{extension} of the partial packing~$\mathcal{J}$.
A partial packing is \emph{$\mathcal{H}$-free} if none of its
members (the boundaried graphs $\{{\bf J}_1, \dots, {\bf
  J}_p\}$) has an $\mathcal{H}$-subdivision. Observe
that if the graphs in ${\cal H}$ are connected, then every partial
subdivision of an $\mathcal{H}$-free partial $\x$-$\mathcal{H}$-packing in
$\bf G$ must contain at least one boundary vertex of~$\bf G$.

\paragraph{Partitions and protrusions.} A {\em rooted
  tree-partition} of a graph $G$ is a triple ${\cal D}=({\cal X},T,s)$
where
$(T,s)$ is a rooted tree and ${\cal X}=\{X_t\}_{t\in V(T)}$ is a
partition of $V(G)$ where either $n(T)=1$ or for every $\{x,y\}\in
E(G)$, there exists an edge $\{t,t'\}\in E(T)$ such that
$\{x,y\}\subseteq X_{t}\cup X_{t'}$ (see also
\cite{Seese,Halin1991203,Ding199645}).  Given an edge $f=\{t,t'\}\in
E(T)$, we define $E_{f}$ as the set of edges of $G$ with one endpoint in
$X_{t}$ and the other in $X_{t'}$.  Notice that all edges in $E_{f}$
are non-loop edges. The {\em width} of ${\cal D}$ is defined as
$$\max\{|X_{t}|_{t\in V(T)}\}\cup \{|E_{f}|_{f\in E(T)}\}.$$
The {\em tree-partition width} of $G$ is 
the minimum width over all tree-partitions of $G$ and will be denoted by $\tpw(G)$. 
The elements of ${\cal X}$ are called~\emph{bags}.

In order to decompose graphs along edge cuts, we introduce the
following edge-counterpart of the notion of
\emph{(vertex-)protrusion} introduced
in~\cite{BodlaenderFLPST09,BodlaenderFLPST09arxiv}.

Given a rooted tree-partition ${\cal D}=({\cal X},T,s)$ of $G$ and a vertex $i\in V(T)$, we 
define $T_i$ as the subtree of $T$ induced by the descendants of $i$ (including $i$) and
\[V_{i} = \bigcup_{h\in
  V(T_{i})}X_{h},\ \ \text{and} \quad G_i  =  G [V_i].\]

Let $W$ be a graph and $t \in \N_{\geq 1}$.
  A  pair ${\bf P}=({\bf G},{\cal D})$ is a {\em $t$-partitioned protrusion} of $W$ 
if there exists an $S\subseteq V(W)$ such that 
\begin{itemize}
\item ${\bf G}=(G,B,\lambda)$ is a  boundaried graph where ${\bf G}\in {\cal B}_{t}$ and ${\bf G}={\bf G}_{S}$ 
for some $S$-splitting $({\bf G}_{S},{\bf G}_{S^{\sf c}})$ of $W$ and
\item ${\cal D}=({\cal X},T,s)$ is a rooted tree-partition of
  $G\setminus B$ of width at most $t$, where $X_{s}$
are the neighbors in $G$ of the vertices in $B$.
\end{itemize}
We say that a $t$-partitioned protrusion $(\mathbf{G}, \mathcal{D})$ (with $\mathcal{D} = (\{X_u\}_{u\in V(T)}, T,s)$)
of a graph $W$ is \emph{$\mathcal{H}$-free} if ${\bf G}$ is $\mathcal{H}$-free.
Recall that in the rooted tree-partition $\mathcal{D}$, for every $u
\in V(T)$, we denote by $V_u$ the union of the bags indexed by
descendants of $u$ (including $u$).
For every vertex $u \in V(T)$, we define the {\em $t$-partitioned
  protrusion  ${\bf P}_{u}$ of $W$} as a pair ${\bf P}_{u}=({\bf G}_{u},{\cal D}_{u})$, 
where ${\cal D}_{u}=(\{X_{v}\}_{v\in V_{T_u}},T_{u},u)$
and ${\bf G}_{u}$ is defined as ${\bf G}_{V_{u}}$
for some $V_{u}$-splitting $({\bf G}_{V_{u}},{\bf G}_{V^{\sf c}_{u}})$
of $W$.
Informally speaking, the rooted tree-partition $\mathcal{D}_u$ is the part
of $\mathcal{D}$ corresponding to the subgraph of $\bf G$ induced by vertices in the bags that are
descendants of~$u$. \autoref{fig:partprot} provides an example of
partitioned protrusions.
We choose the labeling function of ${\bf G}_{s}$ so that it is the same as the one of ${\bf G}$, i.e.,
${\bf G}_{s}={\bf G}$. 
Notice that the labelings of all other ${\bf G}_{u}$'s 
are arbitrary.
For every $u \in V(T)$ we define
\[
\mathcal{G}_u = \{\mathbf{G}_l \mid {l
  \in \children_{(T,s)}(u)} \}.
\]

\begin{figure}[h]
  \centering
  \begin{tikzpicture}
    \draw[rounded corners] (1,0.75) rectangle (2.25,5.25); % W
    \draw[rounded corners, fill = black!20] (2.5,0.75) rectangle (9.75,5.25); % G
    \draw[rounded corners, fill = black!10] (5.5,1) rectangle (9.5,4.5); %G_u
    \begin{scope}[every node/.style = bag]% tree partition
      \draw (3, 2.5) node[label=90:$s$] {$X_s$}
      --++(1, 0) node (b1) {}
      --++(2, 0) node[label=90:$u$] (b3) {$X_u$}
      --++ (1, 1) node (b4) {}
      --++ (1, -1) node (b5) {}
      --++ (1, 0) node{}
      (b5) -- ++(1, 1) node {}
      (b4) -- ++(1, 0) node {}
      (b1) -- ++(0.75, 0.75) node {}
      (b3) --++ (1, -1) node {}
      --++(1, 0) node {};
    \end{scope}
    \draw (1.85, 2.5) ellipse (0.25cm and 0.75cm) node {$B$};
    \draw (1.85, 3) -- (2.8, 2.7);
    \draw (1.85, 2.8) -- (2.8, 2.7);
    \draw (1.85, 2.3) -- (2.8, 2.5);
    \draw (1.85, 3) -- (2.8, 2.5);
    \draw (1.85, 2.1) -- (2.8, 2.4);

    % bottom labels
    \draw[decorate,decoration={brace,amplitude=5pt}] (9.75, 0.5) --
    (5.5, 0.5) node[midway, yshift = -0.4cm] {$\mathbf{G}_u$};
    \draw[decorate,decoration={brace,amplitude=5pt}] (9.75, 0) --
    (2.5, 0) node[midway, yshift = -0.4cm] {$\mathbf{G}$};
    \draw[decorate,decoration={brace,amplitude=5pt}] (9.75, -0.5) --
    (1, -0.5) node[midway, yshift = -0.4cm] {$W$};
    % top labels
    \draw[decorate,decoration={brace,amplitude=5pt}] (5.55, 3.95) --
    (9.45, 3.95) node[midway, yshift = 0.4cm] {$\mathcal{D}_u$};
    \draw[decorate,decoration={brace,amplitude=5pt}] (2.55, 4.7) --
    (9.7, 4.7) node[midway, yshift = 0.4cm] {$\mathcal{D}$};
  \end{tikzpicture}
  \caption{The partitioned protrusion ${\bf P_u} = ({\bf G}_u, \mathcal{D}_u)$ as a substructure
    of $({\bf G}, \mathcal{D})$. White circles
    represent bags of the rooted tree-partitions. Only the edges
    between $B$ and $X_s$ are depicted}
  \label{fig:partprot}
\end{figure}
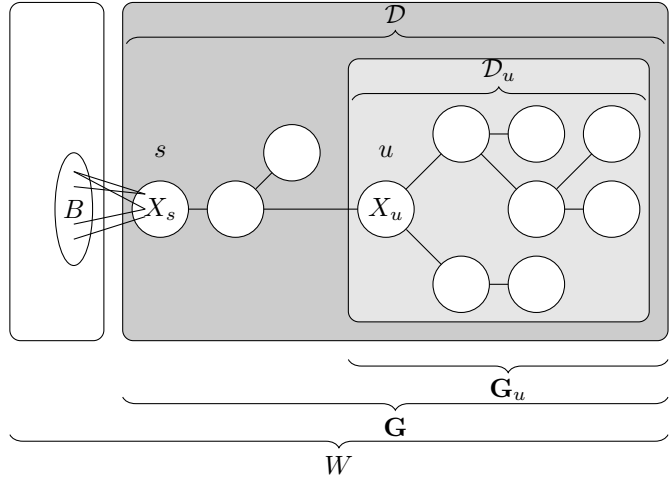

\subsection{Encodings, signatures, and folios}

In this section we introduce tools that we will use to sort
boundaried graphs depending on the subdivisions that are realizable inside.
\paragraph{Encodings.} Let $\mathcal{H}$ be a family of graphs, let
$t \in \N_{\geq 1}$, and let $\x \in \{{\sf v}, {\sf e}\}$.
Recall that if $M$ is an $\mathcal{H}$-subdivision in $G$, we denote by
$\hat{M}$ a graph isomorphic to some graph of $\mathcal{H}$ such that
$M$ is a subdivision of $\hat{M}$ (see \autoref{basicdefs}).
If ${\bf G}=(G,B,\lambda)\in {\cal
  B}_{t}$ is a boundaried graph and $S\subseteq A_{\sf
  x}({\bf G})$, we define $\mu^{\sf x}_{\cal H}({\bf G},S)$ as the
collection of all sets $\{({\bf
  J}_{1},L_{1}),\ldots,({\bf J}_{\sigma},L_{\sigma})\}$ such that
\begin{enumerate}[(i)]
\item $\{{\bf J}_1, \dots, {\bf
  J}_\sigma\}$ is a partial $\x$-$\cal H$-packing of ${\bf G} \setminus S$ of
size~$\sigma$ and
\item $L_i = V(\hat{M_i})\cap V(G)$, where $M_i$ is an extension of
  ${\bf J_i}$, for every $i \in \intv{1}{\sigma}$.
\end{enumerate}

\begin{figure}[ht]
  \centering
  \begin{tikzpicture}[every node/.style = black node, scale = 0.425]
    \begin{scope}[yshift = -11.5cm, xshift = -1cm]
      \node[normal] at (1,0.5) {$\mathcal{H} =  \{\textcolor{colork4}{K_4}, \textcolor{colork23}{K_{2,3}}\}$};
    \end{scope}
    \begin{scope}[yshift = -6cm]
      %% F part
      \draw[rounded corners] (-1,4) -- (-1, -4) -- (-6,-4) .. controls ++(-3,0)  and ++(-3,0) .. (-6, 4) -- cycle;
      \node[normal] at (-7, 0) {$\bf F$};
      % F Boundary
      %
      \draw (-1.5, 3.0) node[fill = colork23] (f0) {};
      \draw (-1.5, 2.0) node[fill = colork4] (f1) {};
      \draw (-1.5, 1.0) node[fill = colork4] (f2) {};
      \draw (-1.5, 0.0) node[fill = colork23] (f3) {};
      \draw (-1.5, -1.0) node[fill = colork23] (f4) {};
      \draw (-1.5, -2.0) node[fill = colork4] (f5) {};
      \draw (-1.5, -3.0) node[fill = colork4] (f6) {};
      %
      % F Boundary end
      \node[black node, fill = colork23] (m1) at (-5,2) {};
      \draw[color = colork23, very thick] (m1) circle (10pt); % branch vertex
      \node[cnode] (c1) at (-2.5, 2.5) {}; % common
      \node[cnode] (c2) at (-2.5, .5) {}; % common
      \draw[color = colork23, very thick] (c2) circle (10pt); % branch vertex
      \node[cnode] (c3) at (-3, -1) {}; % common
      \draw[color = colork23] (f0) -- (c1) -- (m1) -- (c2) -- (f3)
      (m1) .. controls ++(0, -2) and ++(-2, 0) .. (c3) -- (f4);
      \node[black node, fill = colork4] (n1) at (-6, 0) {};
      \draw[color = colork4, very thick] (n1) circle (10pt); % branch vertex
      \node[black node, fill = colork4] (n2) at (-4, 0) {};
      \draw[color = colork4, very thick] (n2) circle (10pt); % branch vertex
      \draw[color = colork4] (n2) -- (n1) .. controls ++(0,3) and ++(-3, 1) .. (c1) -- (f1)
      (n1) .. controls ++(0, -2) and ++(-3,0) .. (f6)
      (n2) -- (c3) .. controls ++(0, -1) and ++(-1,0) .. (f5)
      (n2) -- (c2) -- (f2);
      %% G part 
      \draw[rounded corners] (1,4) -- (1, -4) -- (6,-4) .. controls ++(3,0)  and ++(3,0) .. (6, 4) -- cycle; % G
      \node[normal] at (7, 0) {$\bf G$};
      % G boundary
      % python generated
      \draw (1.5, 3.0) node[fill = colork23] (g0) {};
      \draw[color = colork23, very thick] (g0) circle (10pt); % branch vertex
      \draw (1.5, 1.8) node[fill = colork4] (g1) {};
      \draw (1.5, 0.6) node[cnode] (g2) {};
      \draw (1.5, -0.6) node[fill = colork23] (g3) {};
      \draw (1.5, -1.8) node[fill = colork4] (g4) {};
      \draw (1.5, -3.0) node[fill = colork4] (g5) {};
      %
      % G boundary end
      \node[black node, fill = colork23] (m2) at (4,-1) {};
      \draw[color = colork23, very thick] (m2) circle (10pt); % branch vertex
      \node[cnode] (c4) at (2.5, -1.5) {}; % common
      \draw[color = colork23, very thick] (c4) circle (10pt); % branch vertex
      \draw[color = colork23] (g2) -- (m2) -- (c4) -- (g3)
      (m2) .. controls ++(0,2) and ++(2,0) .. (g0);
      \node[black node, fill = colork4] (n3) at (5, 1) {};
      \draw[color = colork4, very thick] (n3) circle (10pt); % branch vertex
      \node[black node, fill = colork4] (n4) at (4, -2) {};
      \draw[color = colork4, very thick] (n4) circle (10pt); % branch vertex
      \draw[color = colork4] (g4) -- (c4) -- (n4) .. controls ++(-0.25, -0.25) and ++(1, 0) .. (g5)
      (n4) .. controls ++(1, 1) and ++(0, -1) .. (n3)
      (g1) -- (n3) -- (g2); 
      % 
      %%  common
      \draw[color = colork23]
        (f0) -- (g0)
        (f3) -- (g2)
        (f4) -- (g3);
        \draw [color = colork4]
        (f1) -- (g1)
        (f2) -- (g2)
        (f5) -- (g4)
        (f6) -- (g5);
    \end{scope}
  \end{tikzpicture}
  \caption{A member of $\Pkg_{{\sf e},{\cal H}}^{\geq 1}({\bf F}\oplus {\bf G})$. Branch vertices are~circled.}
  \label{fig:folio}
\end{figure}
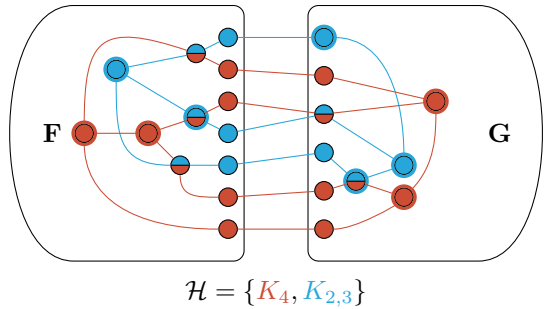

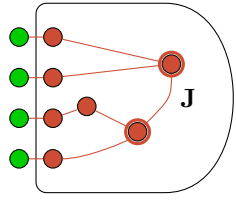
\begin{figure}
  \centering
    \begin{tikzpicture}[every node/.style = black node,scale = 0.45]
      %% G part 
      \draw[rounded corners] (1,2.8) -- (1, -2.8) -- (4.5,-2.8) .. controls ++(3,0)  and ++(3,0) .. (4.5, 2.8) -- cycle; % G
      \node[normal] at (5.5, 0) {${\bf J}$};
      % G boundary
      \draw (1.5, 1.8) node[fill = colork4] (g1) {} ++(-1,0) node[fill = green!80!black] (f1) {};
      \draw (1.5, 0.6) node[fill = colork4] (g2) {} ++(-1,0) node[fill = green!80!black] (f2) {};
      \draw (1.5, -0.6) node[fill = colork4] (g4) {} ++(-1,0) node[fill = green!80!black] (f5) {};
      \draw (1.5, -1.8) node[fill = colork4] (g5) {} ++(-1,0) node[fill = green!80!black] (f6) {};
      %
      % G boundary end
      \node[fill = colork4] (c4) at (2.5, -0.25) {}; % common
      \node[black node, fill = colork4] (n3) at (5, 1) {};
      \draw[color = colork4, very thick] (n3) circle (10pt); % branch vertex
      \node[black node, fill = colork4] (n4) at (4, -1) {};
      \draw[color = colork4, very thick] (n4) circle (10pt); % branch vertex
      \draw[color = colork4] (g4) -- (c4) -- (n4) .. controls ++(-0.25, -0.25) and ++(1, 0) .. (g5)
      (n4) .. controls ++(1, 1) and ++(0, -1) .. (n3)
      (g1) -- (n3) -- (g2); 
        \draw [color = colork4]
        (f1) -- (g1)
        (f2) -- (g2)
        (f5) -- (g4)
        (f6) -- (g5);
        %% subdivision vertices
%        \draw[dashed, rounded corners] (1.125, 2.3) -- (1.125, -2.3) --
%        (1.75, -2.3) -- (3, -0.25) -- (1.75, 2.3) -- cycle;
%        \node[normal] at (2, 2.3) {L};
  \end{tikzpicture}
  \caption{A partial subdivision from the packing of
    \autoref{fig:folio}. Branch vertices are circled.}
  \label{fig:J}
\end{figure}

In other words, $L_i$ contains branch vertices of the partial subdivision
${\bf J}_i$ for every $i \in \intv{1}{\sigma}$ (see Figures~\ref{fig:folio} and~\ref{fig:J}).
The set $\mu^{\sf x}_{\cal H}({\bf G},S)$ encodes all different
restrictions in ${\bf G}$ of partial {\sf x}-${\cal H}$-packings
that avoid the set $S$. Given a boundaried graph ${\bf G} = (G,B,
\lambda)$ and a set $L \subseteq V(G)$ such that every vertex of $V(G)
\setminus L$ has degree 2 in $G$, we define $\kappa(\mathbf{G}, L)$ as
the boundaried graph obtained from $\mathbf{G}$ by dissolving every
vertex of $V(G) \setminus L$, i.e., $\kappa(\mathbf{G}, L) =
(\diss(\mathbf{G}, V(G) \setminus L), B, \lambda)$.
In the definition of $\kappa$ we assume that the boundary vertices of
$\kappa(\mathbf{G}, L)$ remain the same as in ${\bf G}$ while the
other vertices are treated as new vertices (see \autoref{fig:Jcomp}).

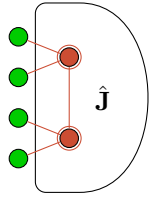
\begin{figure}[h]
  \centering
    \begin{tikzpicture}[every node/.style = black node, scale = 0.45]
      %% G part 
      \draw[rounded corners] (1,2.8) -- (1, -2.8) -- (2,-2.8) .. controls ++(3,0)  and ++(3,0) .. (2, 2.8) -- cycle; % G
      \node[normal] at (3, 0) {$\hat{\bf J}$};
      % G boundary
      \draw (0.5, 1.8) node[fill = green!80!black] (f1) {};
      \draw (0.5, 0.6) node[fill = green!80!black] (f2) {};
      \draw (0.5, -0.6) node[fill = green!80!black] (f5) {};
      \draw (0.5, -1.8) node[fill = green!80!black] (f6) {};
      %
      % G boundary end
      \node[black node, fill = colork4] (n3) at (2, 1.2) {};
      \draw[color = colork4] (n3) circle (10pt); % branch vertex
      \node[black node, fill = colork4] (n4) at (2, -1.2) {};
      \draw[color = colork4] (n4) circle (10pt); % branch vertex
        \draw [color = colork4]
        (f1) -- (n3) -- (f2)
        (f5) -- (n4) -- (f6)
        (n3) -- (n4);
  \end{tikzpicture}
  \caption{The compression of the partial packing of \autoref{fig:J}: $\hat{\bf J} = \kappa({\bf J},L)$.}
  \label{fig:Jcomp}
\end{figure}

This allows us to introduce the following notation aimed at
representing, intuitively, the essential part of each partial packing.
\begin{eqnarray*}
\hat{\mu}^{\sf x}_{\cal H}({\bf G},S) & = &  \{\hat{\cal J}=\{\hat{\bf
  J}_{1},\ldots,\hat{\bf J}_{\sigma}\}=\{\kappa({\bf
  J}_{1}, L_{1}),\ldots,\kappa({\bf J}_{\sigma}, L_{\sigma})\}\mid \\
& &     \ \    \{({\bf J}_{1},L_{1}),\ldots,({\bf J}_{\sigma},L_{\sigma})\}\in\mu_{\cal H}({\bf G},S)\}.
\end{eqnarray*}

\paragraph{Isomorphisms.}
If ${\bf G}=(G,B,\lambda)$ and ${\bf G'}=(G',B',\lambda')$ are two
compatible boundaried graphs in ${\cal B}_{t}$,
 $S\in V(G)$, and $S'\in V(G')$, 
we say that a member $\hat{\cal J}$ of $\hat{\mu}^{\sf x}_{\cal H}({\bf G},S)$
and a member $\hat{\cal J}'$ of $\hat{\mu}^{\sf x}_{\cal H}({\bf G}',S')$
are {\em isomorphic}   if there is a bijection between them such that 
paired elements are isomorphic.
We also say that $\hat{\mu}^{\sf x}_{\cal H}({\bf G},S)$ and $\hat{\mu}^{\sf x}_{\cal H}({\bf G}',S')$
are {\em isomorphic} if there is a bijection between them such that 
paired elements are isomorphic.

We now come to the point where we can define, for every boundaried graph, 
a signature encoding all the possible partial packings that can be
realized in this graph.
\paragraph{Signatures and folios.} 
For  every $y \in \N$, we set
\begin{eqnarray*}
\sig^{\sf x}_{\cal H}({\bf G},y) & = &  \{\hat{\mu}^{\sf x}_{\cal
  H}({\bf G},S),\ S \subseteq A_{\x}(G),\ |S|=y\}
\end{eqnarray*}
and, given two compatible $t$-boundaried graphs ${\bf G}$ and ${\bf
  G}'$ and a $y \in \N$, 
we say that ${\sf sig}^{\sf x}_{\cal H}({\bf G},y)$  and ${\sf sig}^{\sf x}_{\cal H}({\bf G}',y)$ are {\em isomorphic}
if there is a bijection between them such that 
paired elements are  isomorphic.

Finally,  for  $\rho \in \N$, we set
\begin{eqnarray*}
{\sf folio}_{{\cal H},\rho}({\bf G}) & = &  (\sig^{\sf v}_{\cal H}({\bf G},0),\ldots,{\sf sig}^{\sf v}_{\cal H}({\bf G},\rho),{\sf sig}^{\sf e}_{\cal H}({\bf G},0),\ldots,{\sf sig}^{\sf e}_{\cal H}({\bf G},\rho)).
\end{eqnarray*}
Given two $t$-boundaried graphs ${\bf G}$ and ${\bf G}'$,  a
$\rho \in \N$, and a  finite collection of  connected graphs
${\cal H}$, we say that ${\bf G}\simeq_{\mathcal{H},\rho}{\bf G}'$ if
${\bf G}$ and ${\bf G}'$ are compatible, neither ${\bf G}$ nor
${\bf G}'$ contains an $\mathcal{H}$-subdivision, and the elements of ${\sf
  folio}_{{\cal H},\rho}({\bf G})$ and  ${\sf folio}_{{\cal
    H},\rho}({\bf G}')$ are position-wise~isomorphic.

\section{The reduction}
\label{sec:red}

The purpose of this section is to prove the following lemma.

\begin{lemma}\label{l:new_reduce}
There exists a function $\newfun{f:newred} \colon \N^2\rightarrow \N$
and  an algorithm that, given a positive integer $t$, 
 a finite collection ${\cal H}$
of connected graphs where $h=m({\cal H})$, and a $t$-partitioned
protrusion  ${\bf P} = ({\bf G}, (\mathcal{X}, T, s))$ of a graph $W$ with $n({\bf G}) > \funref{f:newred}(h, t)$, outputs either
\begin{itemize}
\item an $\mathcal{H}$-subdivision of $W$ with at most
  $\funref{f:newred}(h, t)$ edges or
\item a graph $W'$ such that\begin{align*}
  \xpack_\mathcal{H}(W') &= \xpack_\mathcal{H}(W),\\
  \xcover_\mathcal{H}(W') &= \xcover_\mathcal{H}(W),\ \text{and}\\
    n(W') &< n(W).  
  \end{align*}
\end{itemize}
Furthermore, this algorithm runs in time $O_{t,h}(n(T))$. Moreover, this algorithm can be enhanced 
so that given a  ${\sf x}$-{${\cal H}$-packing} (resp.  ${\sf x}$-{${\cal H}$-covering}) in $W'$ then it outputs a same size 
${\sf x}$-{${\cal H}$-packing}  (resp.  ${\sf x}$-{${\cal H}$-covering}) in $W$.
\end{lemma}

In other words we can, in linear time, either find a {\sl small}
$\mathcal{H}$-subdivision, or reduce the graph to a smaller one where the
parameters of packing and covering stay the same. 
The main lines of the proof are the following. We consider the
equivalence relation $\simeq_{\mathcal{H},t}$ on $\mathcal{H}$-free
$t$-partitioned protrusions. Informally, $\simeq_{\mathcal{H},t}$
relates the partitioned protrusions in which the same partial packings
are realizable. The first step is to show that this relation has a
finite number of equivalence classes. Recall that each vertex $i$ of the tree
$T$ of the partitioned protrusion $\bf P$ defines a partitioned
protrusion ${\bf P}_i$. We then show that if $T$ is too large (i.e.\ it has
a long path or a vertex of large degree), then for several vertices,
such partitioned
protrusions belong to the same equivalence
class of~$\simeq_{\mathcal{H},t}$. We finally show that in this case,
some of
these vertices can be deleted without changing the packing and covering numbers.

Before giving the proof of \autoref{l:new_reduce}, we need to prove several
intermediate results.
In the sequel, unless stated otherwise, we assume that $\sf{x} \in
\{\sf{v}, \sf{e}\}$, $t\in \N_{\geq 1}$ and that $\mathcal{H}$ is a  finite 
collection of
connected graphs. We set~$h = m(\mathcal{H})$.
Recall that, for a $t$-partitioned protrusion $({\bf G},
(\mathcal{X}, T, s))$ of a graph $W$, and $v\in V(T)$, we denote by
$\mathcal{G}_v$ the set of all boundaried graphs of the form ${\bf
  G}_{V_u}$, for some $V_u$-splitting $({\bf
  G}_{V_u}, {\bf G}_{V_u^{\sf c}})$ of $W$, where $u$ is a child of
$v$. Informally, these are the boundaried graphs induced by the bags
of the subtrees of $T$ rooted at the children of~$v$.

\begin{lemma}\label{l:internal}
  There are two functions $\newfun{f:modbv} \colon \N^2 \to \N$
  and $\newfun{f:intersect} \colon \N^2 \to \N$ such that, for every
  graph $W$ and every $t$-partitioned protrusion $({\bf G},
  (\mathcal{X}, T, s))$ of $W$, if $\mathcal{P}$ is an
  $\mathcal{H}$-free partial $\x$-$\mathcal{H}$-packing in
  $\mathbf{G}$ then:
  \begin{enumerate}[(a)]
  \item \label{e:fa} 
 The partial subdivisions of graphs in ${\cal H}$  that are contained in ${\cal P}$ have in total at most 
 $\funref{f:modbv}(h, t)$ branch~vertices.
  \item \label{e:fb} $\mathcal{P}$ intersects at most $\funref{f:intersect}(h ,t)$
    graphs of~$\mathcal{G}_s$.
  \end{enumerate}
\end{lemma}

\begin{proof}
  \noindent \textit{Proof of \eqref{e:fa}.}
  First, note that any $\mathcal{H}$-free partial $\x$-$\mathcal{H}$-packing in
  $\mathbf{G}$ has cardinality at most~$t,$ because each
  partial subdivision it contains must use a boundary edge of $\mathbf{G}$, and
  two distinct subdivisions of the same packing are 
  (at least) edge-disjoint. %
  Also, each of these partial subdivisions contains at most $\max_{H \in
    \mathcal{H}} n(H) \leq
  h$ branch vertices. Consequently, for every $\mathcal{H}$-free
  partial $\sf x$-$\mathcal{H}$-packing in $\mathbf{G},$ the number of branch vertices
  of graphs of $\mathcal{H}$ it induces in $\mathbf{G}$ is at most
  $t\cdot h$. Hence the function $\funref{f:modbv}(h,t) := t
  \cdot h$ upper-bounds the amount of branch vertices each $\mathcal{H}$-free
  partial packing can contain.

  \noindent \textit{Proof of \eqref{e:fb}.} Let $\zeta$ be the maximum multiplicity of an
  edge in a graph of~$\mathcal{H}$. Because of \eqref{e:fa}, every $\mathcal{H}$-free partial
  $\x$-$\mathcal{H}$-packing $\mathcal{P}$ in ${\bf G}$ has at most
  $\funref{f:modbv}(h, t)$ branch vertices of graphs of
  $\mathcal{H}$, so at most $\funref{f:modbv}(h, t)$
  graphs of $\mathcal{G}_s$ may contain such vertices.
  Besides, $\mathcal{P}$ might also contain paths free of branch
  vertices linking pairs of branch vertices. Since there are at most $(\funref{f:modbv}(h, t))²$ such pairs and no pair
  will need to be connected with more than $\zeta \leq h$
  distinct paths, it follows that at most $(\funref{f:modbv}(h, t))² \cdot  h$
  graphs of 
  $\mathcal{G}_s$ contain vertices from these paths. Therefore, every
  $\mathcal{H}$-free partial $\sf x$-$\mathcal{H}$-packing intersects at most
  $\funref{f:modbv}(h, t) +
 (\funref{f:modbv}(h, t))² \cdot h =:
  \funref{f:intersect}(h, t)$ graphs of~$\mathcal{G}_s.$
\end{proof}

\begin{lemma}\label{l:muhat_finite}
  There is a function $\newfun{f:simclas} \colon \N^2 \to \N$ such that the image of the function $\hat{\mu}^{\sf x}_{\mathcal H}$, when its
  domain is restricted to
\[
\{(\mathbf{G},S),\ \mathbf{G}\ \text{is}\ \mathcal{H}\text{\rm -free } {\rm and}\ S \subseteq A_{\sf x}(\mathbf{G})\},
\]
  has size upper-bounded by $\funref{f:simclas}(h,t)$.
\end{lemma}

\begin{proof}
  Let $\mathbf{G}$ be an $\mathcal{H}$-free $t$-boundaried graph and let $S \subseteq
  A_{\sf x}(\mathcal{G}).$ 
  From \autoref{l:internal}(\ref{e:fa}), every $\mathcal{H}$-free partial $\sf
  x$-$\mathcal{H}$-packing in $\mathbf{G}$ contains at most
  $\funref{f:modbv}(h,t)$ branch vertices.
  This partial packing can in addition use at most $t$ boundary vertices.
  Let $\mathcal{C}_{h,t}$ be the class of all $(\leq t)$-boundaried
  graphs on at most $\funref{f:modbv}(h,t) + t$ vertices. Clearly
  the size of this class is a function depending on~$h$ and~$t$ only.
  Recall that the elements of the set $\hat{\mu}^{\sf x}_{\mathcal
    H}(\mathbf{G}, S)$ are obtained from partial $\sf
  x$-$\mathcal{H}$-packings by dissolving internal vertices of the
  paths linking branch vertices, hence every element of $\hat{\mu}^{\sf x}_{\mathcal
    H}(\mathbf{G}, S)$  is a $t$-boundaried graph of
  $\mathcal{B}_t$ having at most $\funref{f:modbv}(h,t)
  + t$ vertices. Therefore, for any $\mathcal{H}$-free $t$-boundaried graph $\mathbf{G}$ and
  subset~$S \subseteq A_{\sf x}(\mathbf{G})$, we have $\hat{\mu}^{\sf
    x}_{\mathcal H}(\mathbf{G}, S) \subseteq \mathcal{C}_{h,t}$.
  As a consequence, the image of the function $\hat{\mu}^{\sf x}_{\mathcal
    H}$ when restricted to $\mathcal{H}$-free $t$-boundaried
   graphs $\mathbf{G} \in \mathcal{B}_t$ (and
  subsets~$S \subseteq A_{\sf x}(\mathbf{G})$) is a subset of the
  power set of $\mathcal{C}_{h,t}$, so its size is upper-bounded by a
  function (which we call~$\funref{f:simclas}$) that depends only on $h$ and~$t$.
\end{proof}

\begin{corollary}\label{c:finite-ec}
  There is a function $\newfun{f:eqclass} \colon \N^2
  \to \N$ such that the relation $\simeq_{\mathcal{H},t}$
  partitions $\mathcal{H}$-free $t$-boundaried graphs into at
  most $\funref{f:eqclass}(h,t)$ equivalence~classes.
\end{corollary}

 The following lemma follows directly from the definition of~$\hat{\mu}^{\sf x}_\mathcal{H}$.
\begin{lemma}\label{l:diss-pack}
  Let $\mathbf{F}, \mathbf{G} \in \mathcal{B}_t$ be two compatible~$t$-boundaried graphs and let~$k \in \N$. Then we have:

  \[\Pkg_{{\sf x},{\cal H}}^{\geq k}(\mathbf{F} \oplus \mathbf{G}) \neq
  \emptyset \iff \exists \hat{\mathcal{J}} \in \hat{\mu}^{\sf x}_{\mathcal
    H}(\mathbf{G}, \emptyset),\ \Pkg_{{\sf x},{\cal H}}^{\geq
    k}(\mathbf{F} \oplus \cupall \hat{\mathcal{J}}) \neq
  \emptyset.\]
\end{lemma}

The choice of the definition of the relation $\simeq$ is justified  by
the following lemma. Roughly speaking, it states that we can replace a
$t$-partitioned protrusion of a graph with any other
$\simeq_{\mathcal{H},t}$-equivalent $t$-partitioned protrusion without
changing the covering and packing number of the graph. The reduction
algorithm  that we give after this lemma  relies on this powerful property.

\begin{lemma}[protrusion replacement]\label{l:protrep}
  Let $\mathbf{F}, \mathbf{G}, \mathbf{G'} \in \mathcal{B}_t$ be three
  compatible boundaried graphs such that $\mathbf{G} \simeq_{\mathcal{H},
    t} \mathbf{G'}$. For every $k \in \N$, we have:
  \begin{enumerate}[(i)]
  \item \label{e:c1} $\Pkg_{{\sf
        x},{\cal H}}^{\geq k}(\mathbf{F} \oplus \mathbf{G}) \neq \emptyset \iff
    \Pkg_{{\sf x},{\cal H}}^{\geq k}(\mathbf{F} \oplus \mathbf{G'}) \neq \emptyset$ and
  \item \label{e:c2} $\Cvg_{{\sf x},{\cal H}}^{\leq k}(\mathbf{F} \oplus \mathbf{G}) \neq \emptyset
    \iff \Cvg_{{\sf x},{\cal H}}^{\leq k}(\mathbf{F} \oplus \mathbf{G'}) \neq \emptyset$.
  \end{enumerate}
\end{lemma}

\begin{proof}
  Notice that, by definition of $\simeq_{\mathcal, t}$, the boundaried
  graphs $\mathbf{G}$ and $\mathbf{G'}$ are $\mathcal{H}$-free.

  \noindent \textit{Proof of item (\ref{e:c1}), ``$\Rightarrow$''.}
  Let $\mathcal{M}$ be an $\x$-$\mathcal{H}$-packing of size at least
  $k$ in $\mathbf{F} \oplus \mathbf{G}$, whose set of branch vertices is $L$. We define
  \begin{eqnarray*}
 {\bf J}_{\mathbf{F}} & = &  (\cupall \mathcal{M}) \cap_{{\bf F} \oplus {\bf
      G}} {\bf F},\\
{\bf J}_{\mathbf{G}} & = &  (\cupall \mathcal{M})
  \cap_{{\bf F} \oplus {\bf G}} {\bf G}, and\\
  \hat{\mathbf{J}}_{\mathbf{G}}  & = & \bigcup_{{M}\in {\cal M}}\kappa(M\cap_{{\bf F} \oplus {\bf G}} {\bf G},
  L\cap V(G)).%(V(M)\cap V(G))).
  \end{eqnarray*}
   Note that $\hat{\mathbf{J}}_{\mathbf{G}} \in \hat{\mu}^{\sf x}_{\mathcal
      H}(\mathbf{G}, \emptyset)$ and that ${\bf F} \oplus \hat{\mathbf{J}}_{\mathbf{G}} $ has an $\x$-$\mathcal{H}$-packing of size at
    least~$k$ (cf. \autoref{l:diss-pack}). 
    By definition of $\simeq,$ there is a bijection $\psi$ between
    $\hat{\mu}^{\sf x}_{\mathcal H}(\mathbf{G}, \emptyset)$ and
    $\hat{\mu}^{\x}_{\mathcal H}(\mathbf{G'}, \emptyset)$. Let
    $\hat{\mathbf{J}}_{\mathbf{G}}'$ be the image of
    $\hat{\mathbf{J}}_{\mathbf{G}}$ by $\psi.$ 
    Since $\hat{\mathbf{J}}_{\mathbf{G}}'$ and $\hat{\mathbf{J}}_{\mathbf{G}}$
    are isomorphic, ${\bf F} \oplus \hat{\mathbf{J}}_{\mathbf{G}} '$ also has an
    $\x$-$\mathcal{H}$-packing of size at least $k$. By
    \autoref{l:diss-pack}, this implies that such a packing exists
    in ${\bf F} \oplus {\bf G'}$ as well. The direction
    ``$\Leftarrow$'' is symmetric as $\bf G$ and ${\bf G'}$ play the
    same role.

    \noindent \textit{Proof of item (\ref{e:c2}), ``$\Rightarrow$''.} Let $C \subseteq
    A_{\x}(\mathbf{F} \oplus \mathbf{G})$ be a minimum
    $\x$-$\mathcal{H}$-covering of $\bf F\oplus G$ of size at most
    $k$. Let $S = C \cap A_\x({\bf G})$. Since we assume that
    $\mathbf{G}$ is $\mathcal{H}$-free and that $C$ is minimum, we can
    also assume that $|S| \leq t$ (otherwise we could get a smaller covering by taking the $t$ boundary vertices/edges of $\bf G$).
    By our assumption that ${\bf
      G} \simeq_{\mathcal{H}, t} {\bf G'}$, there is an isomorphism
    between $\sig^\x_{\mathcal{H}}({\bf G}, |S|)$ and
    $\sig^\x_{\mathcal{H}}({\bf G'}, |S|)$. Let $S' \subseteq 
    A_\x({\bf G'})$ be a set such that $\hat{\mu}^{\x}_{\mathcal
      H}(\mathbf{G}, S)$ is sent to $\hat{\mu}^{\x}_{\mathcal H}(\mathbf{G'}, S')$ by this  
    isomorphism.
    Then observe that every partial packing $\mathcal{J}'$ of
    %${\mu}^{\x}_{\mathcal H}({\bf G}', S')$
    ${\bf G}' \setminus S'$,
    such that $\left (\mathbf{F} \setminus C \right ) \oplus
    (\cupall \mathcal{J}')$ has an $\mathcal{H}$-subdivision, can be translated into a
    partial packing $\mathcal{J}$ of
    % ${\mu}^{\x}_{\mathcal H}({\bf G}, S)$
    ${\bf G} \setminus S$
    such that $\left (\mathbf{F} \setminus C \right )\oplus
    (\cupall\mathcal{J})$ also has such a subdivision, in the same way as in the
    proof of item (\ref{e:c1}) above. As
    $C$ is a cover, this would lead to contradiction.
    Therefore ${\mu}^{\x}_{\mathcal H}({\bf G}, S)$ does not contain such a partial packing. As a
    consequence, $\left ( C \cap A_\x({\bf F})\right ) \cup S'$ is a covering of ${\bf F}
    \oplus {\bf G'}$ of size at most $k$.
    As in the previous case, the proof of direction ``$\Leftarrow$''
    comes from the symmetry in the statement.
\end{proof}

\autoref{l:protrep} can be rewritten as follows.

\begin{corollary}\label{c:equiv}
  Under the assumptions of \autoref{l:protrep}, we have $\xpack({\bf F}
  \oplus {\bf G}) = \xpack({\bf F} \oplus {\bf G'})$ and $\xcover({\bf F} \oplus {\bf G}) =
  \xcover({\bf F} \oplus {\bf G'}).$
\end{corollary}

%Let $\newfun{f:maxdeg}(h,t) = \funref{f:eqclass}(h,
%t) \cdot  \funref{f:intersect}(h, t)$ and
%$\newfun{f:maxheight}(h, t) = \left (\funref{f:maxdeg}(h, t) \right )^{\funref{f:eqclass}(h, t)}$.
% replaced maxheight with eqclass

Recall that $\funref{f:eqclass}(h,
t)$ denotes the number of equivalence classes of
$\simeq_{\mathcal{H},t}$ among boundaried graphs of~$\mathcal{B}_t$.
For every $h,t\in \N$, let
\[\newfun{f:maxdeg}(h,t) = \funref{f:eqclass}(h,
t) \cdot  \funref{f:intersect}(h, t)\quad \text{and}\quad \funref{f:newred}(h, t) = 2ht^3\cdot \left (\funref{f:maxdeg}(h, t) \right )^{\funref{f:eqclass}(h, t)+1}.\]
Let us give some intuition about these definitions. The first remark
is an application of the pigeonhole principle.

\begin{observation}\label{nr1}
  In a collection of more than $\funref{f:maxdeg}(h,t)$
  $\mathcal{H}$-free boundaried graphs of $\mathcal{B}_t$, there
  is one that is equivalent (w.r.t.\ $\simeq_{\mathcal{H}, t}$) to
  $\funref{f:intersect}(h, t)$ other graphs of the collection.
\end{observation}
\begin{observation}\label{nr2}
If $(T, s, \mathcal{X})$ is a rooted tree-partition of a graph $G$ with the
following properties:
\begin{itemize}
\item $(T, \mathcal{X})$ has width at most $t$;
\item $T$ has height at most $\funref{f:eqclass}(h,
t)$; and
\item $T$ has degree at most $\funref{f:maxdeg}(h, t)+1$,
\end{itemize}
then $G$ has at most
%$\funref{f:maxheight}(h, t)$ edges.
$\funref{f:newred}(h,t)$ vertices, and every $\mathcal{H}$-subdivision of $G$ has at most $\funref{f:newred}(h,t)$~edges.
\end{observation}

\begin{proof}
The above assumptions imply that $T$ has at most
$\left (\funref{f:maxdeg}(h, t) \right
)^{\funref{f:eqclass}(h, t)+1}$ vertices.
Every bag of $(T, \mathcal{X})$ contains at most $t$ vertices of $G$,
therefore $G$ has at most $\left (\funref{f:maxdeg}(h, t)
\right )^{\funref{f:eqclass}(h, t)+1} \cdot t \leq
\funref{f:newred}(h,t)$ vertices.
Also, every bag induces a subgraph with at most $t(t-1)/2$ multiedges (i.e.\ without counting multiplicities), and for every edge $f$ of
$T$ we have $|E_f| \leq t$, hence every bag contributes for at most $t^2 + t$ multiedges. Therefore $G$ has at most $(t+t^2)\left (\funref{f:maxdeg}(h, t) \right
)^{\funref{f:eqclass}(h, t)+1}$ multiedges.
Now, observe that an edge of $G$ is used at most $h$ times by an $\mathcal{H}$-subdivision (if it is a multiple edge), since every path connecting two branch vertices of a subdivision uses a given edge at most once. We deduce that an $\mathcal{H}$-subdivision of $G$ contains at most $h \cdot (t+t^2)\left (\funref{f:maxdeg}(h, t) \right
)^{\funref{f:eqclass}(h, t)+1} \leq \funref{f:newred}(h,t)$~edges.
\end{proof}

The two next lemmas are the main tools used in the proof
of~\autoref{l:new_reduce}. Under different conditions, they provide either a small
subdivision, or a reduced graph. \autoref{l:bdg} considers
the case where a vertex of $T$ has high degree, whereas
\autoref{l:lpth} deals with the situation where $T$ has a long path.

\begin{lemma}\label{l:bdg}
 Let ${\bf P} = ({\bf G}, (\mathcal{X}, T, s))$ be a $t$-partitioned
  protrusion of a graph $W$, and let $u\in V(T)$ be a vertex with
  more than $\funref{f:maxdeg}(h, t)$ children such that
  for every $v \in \children_{(T,s)}(u)$,  we have $m({\bf G}_v) \leq
  \funref{f:eqclass}(h, t)$. Then, either
  \begin{itemize}
  \item $W$ contains  an $\mathcal{H}$-subdivision $M$ with at most $\funref{f:eqclass}(h, t)$ edges or
  \item there exists a graph $W'$ such that
 \begin{align*}
  \xpack_\mathcal{H}(W') &= \xpack_\mathcal{H}(W),\\
  \xcover_\mathcal{H}(W') &= \xcover_\mathcal{H}(W),\ \text{and}\\
    n(W') &< n(W).  
  \end{align*}  \end{itemize}
  Moreover, there is an algorithm that, given such ${\bf P},W,$ and $u$,  returns
  either $M$ or $W'$ as above in  $O_{h, t}(1)$ steps.
\end{lemma}

\begin{proof}
  As $u$ has more than $\funref{f:maxdeg}(h,t)$ children,
  it contains a collection of $d = \funref{f:intersect}(h,
  t) +1$ children $v_1, \dots, v_d$, such that ${\bf G}_{v_1}
  \simeq_{\mathcal{H}, t} {\bf G}_{v_i}$ for every $i \in
  \intv{2}{d}$ (by \autoref{nr1}). 
Let us now assume that ${\bf G}_u$ is $\mathcal{H}$-free.  Since every $\x$-$\mathcal{H}$-packing of $W$ will intersect at most
  $\funref{f:intersect}(h, t)$ bags of children of $u$ (by \autoref{l:internal}(\ref{e:fb})), we can safely
  delete one of the $\funref{f:intersect}(h, t) +1$
  equivalent subgraphs mentioned above. We use the following algorithm
  in order to find such a bag to delete or a small $\mathcal{H}$-subdivision.
  \begin{enumerate}
  \item Let $A$ be an array of $\funref{f:eqclass}(h, t)$ counters initialized to
    0, each corresponding to a distinct equivalence class
    of~$\simeq_{\mathcal{H}, t}$;
  \item Pick a vertex $v \in \children_{(T,s)}(u)$ that has not been
    considered yet;\label{e:pick}
  \item If $G_v$ contains an $\mathcal{H}$-subdivision $\mathcal{M}$, then
    return $\mathcal{M}$ and exit;\label{e:modcheck}
  \item Otherwise, increment the counter of $A$ corresponding to the
    equivalence class of ${\bf G}_u$ by one;\label{e:folcomp}
  \item \label{e:end} If this counter reaches $d+1$, return $v$, otherwise go back to
    Line~\ref{e:pick}.
  \end{enumerate}

  Notice that the subdivision returned in Line~\ref{e:modcheck} has size at
  most $t\cdot \funref{f:eqclass}(h, t)$ (as we assume that $m(G_v) \leq
  \funref{f:eqclass}(h, t)$) and that the vertex returned in Line~\ref{e:end} has the desired property. By
  \autoref{c:finite-ec}, the relation
  $\simeq_{\mathcal{H}, t}$ has at most $\funref{f:eqclass}(h, t)$
  equivalence classes, thus the main loop will be run at most $\funref{f:eqclass}(h, t)
  \cdot \funref{f:intersect}(h, t) + 1$ times (by the
  pigeonhole principle). Eventually, Lines~\ref{e:modcheck}
  and~\ref{e:folcomp} can be performed in $O_{h, t}(1)$-time given
  that $G_v$ has size bounded by a function of~$h$ and~$t$.

    In the end, we return $W' = W \setminus V(G_v)$ if the algorithm
    outputs $v$ and $\mathcal{M}$ otherwise.
\end{proof}

\begin{lemma}\label{l:lpth}
  There is an algorithm that, given a $t$-partitioned protrusion ${\bf P} = ({\bf G}, (\mathcal{X},
  T, s))$ of a graph $W$ and a vertex $u \in
  V(T)$ such that 
  \begin{itemize}
  \item $u$ has height exactly $\funref{f:eqclass}(h,t)$ in~$(T,s)$,
  \item the graph of ${\bf G}_u$ is $\mathcal{H}$-free, and
  \item $T_u$ has maximum degree at most $\funref{f:maxdeg}(h, t)$,
  \end{itemize}
  outputs a graph $W'$ such that
  \begin{align*}
  \xpack_\mathcal{H}(W') &= \xpack_\mathcal{H}(W),\\
  \xcover_\mathcal{H}(W') &= \xcover_\mathcal{H}(W),\ \text{and}\\
    n(W') &< n(W).  
  \end{align*}
  Moreover, this algorithm runs in $O_{h, t}(1)$-time.
\end{lemma}

\begin{proof}
As in the proof of \autoref{nr2}, we use the fact that every $\mathcal{H}$-subdivision of ${\bf G}_u$ uses every (multi)edge at most $h$ times. A consequence is that the boundaried subgraph of ${\bf G}_u$ obtained by setting the multiplicity of every edge $e$ to the minimum of $h$ and the multiplicity of $e$ contains an $\mathcal{H}$-subdivision iff ${\bf G}_u$ does.
As the number of vertices and edges of this subgraph is bounded by a function of $h$ and $t$, we can therefore check in $O_{h,t}(1)$-time if ${\bf G}_u$ contains an $\mathcal{H}$-subdivision. If one is found, it has at most $\funref{f:newred}(h,t)$ edges (\autoref{nr2}) and we are done.

Let us now consider the case where ${\bf G}_u$ is $\mathcal{H}$-free.
  By definition of the vertex $u$, there is a path of
  $\funref{f:eqclass}(h,t) + 1$ vertices from a leaf of $T_u$ to 
  $u$. Let us arbitrarily choose, for every vertex $v$ of this path, a
  $V_v$-splitting $({\bf G}_{G_v}, {\bf G}_{G_v^c})$ of~$\bf G$.
  By definition of $\funref{f:eqclass}(h,t)$ (the number of equivalence classes in $\simeq_{\mathcal{H},t}$ in $\mathcal{B}_t$), there
  are two distinct vertices $v,w$ on this path such that ${\bf G}_v
  \simeq_{\mathcal{H}, t} {\bf G}_w$.
  As mentioned above, the number
  edges of ${\bf G}_u$ is bounded by a function of $h$ and $t$, hence finding these two 
  vertices can be done in $O_{h, t}(1)$-time.  
  Let us assume
  without loss of generality that $w$ is a descendant of $v$.
  Let ${\bf H}$ be the boundaried graph such that $W = {\bf H}
  \oplus {\bf G}_{G_v}$ and let $W'  = {\bf H} \oplus {\bf
    G}_{G_w}$. By \autoref{c:equiv}, we have
  $\xpack_\mathcal{H}(W') = \xpack_\mathcal{H}(W)$ and
  $\xcover_\mathcal{H}(W') = \xcover_\mathcal{H}(W)$. Furthermore, the graph $W'$
  is clearly smaller than $W$.
\end{proof}

%For every $h,t\in \N^2$, let $\funref{f:newred}(h,t) = t\cdot \funref{f:eqclass}(h, t)$.
We are now ready to prove \autoref{l:new_reduce}.

\begin{proof}[Proof of \autoref{l:new_reduce}]
  Observe that since $n(G) > \funref{f:newred}(h, t)$ and each bag of $(\mathcal{X}, T, s)$ contains at
most $t$ vertices, we have
\[n(T) > \funref{f:newred}(h, t)/t = \left (\funref{f:eqclass}(h, t) \cdot 
  \funref{f:intersect}(h, t) \right
)^{\funref{f:eqclass}(h,t)}.\]
Therefore, $T$ has either diameter more than $\funref{f:eqclass}(h,t)$ or a vertex of
degree more than $\funref{f:eqclass}(h, t) \cdot 
\funref{f:intersect}(h, t)$.

Let us consider the following procedure.
\begin{enumerate}
\item By a DFS on $(T,s)$, compute the height of each vertex of $T$
  and find (if it exists) a vertex $v$ of degree more than $\funref{f:maxdeg}(h, t)+1$ and height at most $\funref{f:eqclass}(h,t)-1$ that
  has minimum height.\label{e:bfs0}
\item If such a vertex $v$ is found, then apply the algorithm of
  \autoref{l:bdg} on ${\bf P}$ and~$v$, and return the obtained~result.
\item Otherwise, find a vertex $u$ of height exactly 
  $\funref{f:eqclass}(h,t)$ in $(T,s)$ and then apply the algorithm of
  \autoref{l:lpth} on $\bf P$ and $(T_u, u)$ and return the obtained result.\label{hei}
\end{enumerate}
  Observe that since $n(G) > \funref{f:newred}(h, t)$,
  \autoref{nr2} implies that either $T$ has diameter more
  than $\funref{f:eqclass}(h,t)$, or it contains a vertex of
degree more than $\funref{f:maxdeg}(h, t)+1$. Therefore, the
vertex $u$ of line~\ref{hei} always exists in the case where
no vertex of high degree is found in line~\ref{e:bfs0}.
The correctness of this algorithm follows from 
\autoref{l:bdg} and \autoref{l:lpth}. The DFS done in the first
step takes time $O(n(T))$ and the rest of the algorithm takes time
$O_{h, t}(1)$ according to the aforementioned lemmata.

The fact that a ${\sf x}$-{${\cal H}$-packing}
of $W'$ can be lifted to an equal size ${\sf x}$-{${\cal H}$-packing}
of $W$ follows easily when \autoref{l:bdg} is applied, as then $W'$ is a subgraph of $W$.
In the case that  \autoref{l:lpth} is applied, the new ${\sf x}$-{${\cal H}$-packing}
is obtained because of \autoref{c:equiv} and the supporting protrusion replacement \autoref{l:protrep}, whose proof 
indicates which part of the ${\sf x}$-{${\cal H}$-packing} of $W'$
should be replaced with what in order to obtain an equal size  ${\sf x}$-{${\cal H}$-packing} of~$W$.
\end{proof}
%
%The proof of \autoref{l:thered} is now straighforward.
%
%\begin{proof}[Proof of \autoref{l:thered}]
%Let ${\sf x}\in\{{\sf v},{\sf e}\}$. We describe here the algorithm
%mentioned in the statement of \autoref{l:thered}.
%We are given a graph collection ${\cal H}$ where $m({\cal H})\leq h$, a
%graph $W,$ and a $t$-edge protrusion $Y$ of $W$ with extension more than $\funref{f:newred}(h,t)$ and rooted tree-partition $((\mathcal{X}, T,s))$. Let
%us consider the $t$-partitioned protrusion ${\bf P} = ({\bf G}_{Y},
%(\mathcal{X}, T, s))$ of $W$, for an arbitrarily chosen $({\bf G}_{Y},
%{\bf G}_{Y^c})$-splitting of $W$. Observe that such a protrusion can be
%constructed in $O_t(1)$-time from $Y$ and $(\mathcal{X}, T,s)$.
%As $Y$ has extension more than $\funref{f:newred}(h,t)$,
%we get that $n(G) > \funref{f:newred}(h,t)$. Applying the algorithm of
%\autoref{l:new_reduce} on ${\bf P}$ and $W$ return the desired
%result in time $O_{h,t}(n(T)) = O_{h}(n(G))$.
%\end{proof}

\section{From the \texorpdfstring{Erdős-Pósa}{Erdös-Posa property} property to approximation}
\label{hki6yqpo9}

For the purposes of this section we define $\Theta_{r}={\sf ex}(\theta_{r})$. 
We need the following that is one of the main results in~\cite[Theorem~3.2]{ChatzidimitriouRST15mino}.

\begin{proposition}
\label{rk6ter}
There is an algorithm that, with input three positive integers $r\geq 2, w, z$ and a connected $m$-edge graph $W$, where $m\geq z>r\geq 2$,
outputs one of the following:
\begin{itemize} 
\item a $\Theta_{r}$-subdivision of $W$ with at most $z$ edges, 
\item a  $(2r-2)$-partitioned protrusion $({\bf G},{\cal D})$ of $W$, 
where ${\bf G}=(G,B,\lambda)$ and such that $G$ is a connected graph 
and $n({\bf G})>w,$ or
\item an $H$-minor model of $W$ for some graph $H$ with $\delta(H)\geq \frac{1}{r-1}2^{\frac{z-5r}{4r(2w+1)}}$,
\end{itemize}
in $O_{r}(m)$ steps.
\end{proposition}

\subsection{A lemma on reduce or progress}

The proof of the next lemma combines 
\autoref{rk6ter} and \autoref{l:new_reduce}.
\begin{lemma}[Reduce or progress]
    \label{l:eppump}
  There is an algorithm that, with input  ${\sf x}\in\{{\sf v},{\sf
    e}\}$, $r \in \N_{\geq 2}$, $k \in \N$ and an $n$-vertex graph $W$,
  outputs one of the following:
  \begin{itemize} 
  \item a $\Theta_{r}$-subdivision of $W$ with at most $O_{r}(\log k)$ edges;
  \item  a graph $W'$ where
    \begin{align*}
      \xcover_{\cal H}(W') &= \xcover_{\cal H}(W),\\
      \xpack_{\cal H}(W') &=\xpack_{\cal H}(W),\ \text{and}\\
      n(W') &< n(W); or
    \end{align*}
  \item an $H$-minor model in $W$, for some graph $H$ with $\delta(H)\geq k(r+1)$,
  \end{itemize}
 in $O_{r}(m)$ steps. Moreover, this algorithm can be enhanced 
so that, when the second case applies, given a  ${\sf x}$-{${\cal H}$-packing} 
(resp.  ${\sf x}$-{${\cal H}$-covering}) in $W'$ then it outputs a same size 
${\sf x}$-{${\cal H}$-packing}  (resp.  ${\sf x}$-{${\cal H}$-covering}) in $W$.
\end{lemma}

\begin{proof}
We set $t = 2r-2$, $w = \funref{f:newred}(h,
  t)$, $z = 2r(w-1)\log(k(r + 1)(r - 1)) + 5r$, and $h = m(\Theta_r)$. Observe that
  $z = O_r(\log k)$ and $h,t,w = O_r(1)$. Also observe that our choice
  for variable~$z$ ensures that~$2^\frac{z-5r}{2r(w-1)}/(r-1) = k(r+1)$.\smallskip
  
  By applying the algorithm of \autoref{rk6ter} to $r,w,z$, and $W$,
  we obtain in~$O_r(m(W))$-time either:
  \begin{itemize} 
  \item a $\Theta_{r}$-subdivision in $W$ of at most $z$ edges (first case),
  \item a $(2r-2)$-edge-protrusion $Y$ of $W$ with extension $> w$ (second case), or
  \item an $H$-minor model $M$ in $W$, for some graph $H$ with
    $\delta(H)\geq k(r+1)$ (third case).
  \end{itemize}

  In the first case, we return the obtained $\Theta_r$-subdivision.\smallskip
  In the second case, by applying the algorithm of \autoref{l:new_reduce} on $Y$, we
  get in $O(n(W))$-time either a $\Theta_{r}$-subdivision of $W$ on at most $w = O_r(1)$
  vertices, or a graph  $W'$ where, for ${\sf x}\in\{{\sf v},{\sf e}\}$,\  
    $\xcover_{\cal H}(W')=\xcover_{\cal H}(W)$, $\xpack_{\cal
      H}(W')=\xpack_{\cal H}(W)$ and $n(W')<n(W)$.\smallskip
  In the third case, we return the minor model $M$.\smallskip\smallskip

    In each of the above cases, we get after $O(m)$ steps either a minor model of a graph with
    minimum degree more than $k(r+1)$, a $\Theta_{r}$-subdivision in
    $W$ with at most $z$ edges, or an equivalent graph of smaller size.
\end{proof}

It might not be clear yet to what purpose the minor model of a graph of
degree more than $k(r+1)$ output by the algorithm of
\autoref{l:eppump} can be used. An answer is given by the following
lemmata, which state that such a graph contains a packing of at least
$k$ minor models of~$\Theta_r$. These lemmata will be used in the design of
the appoximation algorithms in \autoref{s:appx}.

\begin{lemma}\label{l:bigepack}%\label{l:extract}
  There is an algorithm that, given $k,r \in \N_{\geq 1}$ and a graph
  $G$ with $\delta(G) \geq kr$, returns a member of ${\bf P}_{{\sf
      e},\Theta_{r}}^{\geq k}(G)$ in $G$ in $O(m)$ steps.
\end{lemma}

\begin{proof}
Starting from any vertex $u$, we grow a maximal path $P$ in $G$ by iteratively
adding to $P$ a vertex that is adjacent to the previously added
vertex but does not belong to~$P$. Since $\delta(G) \geq kr$, any
such path will have length at least $kr+1$.
At the end, all the neighbors of the last vertex $v$ of $P$ belong to $P$
(otherwise $P$ could be extended). Since $v$ has degree at least $kr$,
$v$ has at least $kr$ neighbors in $P$. Let $w_0, \dots, w_{kr-1}$ be
an enumeration of the $kr$ first neighbors of $v$ in the order given
by $P$, starting from $u$. For every $i \in \intv{0}{k-1}$, let $S_i$
be the subgraph of $G$ induced by $v$ and the subpath of $P$ starting at
$w_{ir}$ and ending at $w_{(i+1)r-1}$. Observe that for every $i \in \intv{0}{k-1}$,
$S_i$ contains a $\Theta_r$-subdivision and that the intersection of every
pair of graphs from $\{S_i\}_{i \in \intv{0}{k-1}}$ is $\{v\}$. Hence $P$ contains a member of ${\bf P}_{{\sf
    e},\Theta_{r}}^{\geq k}(G)$, as desired. Every edge of $G$ is considered at most
once in this algorithm, yielding to a running time of $O(m)$~steps. 
\end{proof}

\begin{corollary}\label{c:extract}%\label{c:bigepack}
  There is an algorithm that, given $r\in \N_{\geq 1}$ and a graph $G$
  with $\delta(G) \geq r$, returns a $\Theta_r$-subdivision in $G$ in~$O(m)$-steps.
\end{corollary}
Observe that the previous lemma only deals with edge-disjoint packings.
An analogue of \autoref{l:bigepack} for vertex-disjoint packings
can be proved using \autoref{p:bazgan}, to the price of a worse time complexity.

\begin{proposition}[Theorem 12 of~\cite{Bazgan2007979}]\label{p:bazgan}
  Given $k,r \in \N_{\geq1}$ and an input graph $G$ such that
  $\delta(G)\geq k(r + 1) - 1$, a partition $(V_1, \dots, V_k)$ of
  $V(G)$ satisfying $\forall i \in \intv{1}{k},\ \delta(G[V_i]) \geq
  r$ can be found in $O(n^{c})$ steps for some~$c\in \N$.
\end{proposition}

\begin{lemma}\label{l:vpackex}
There is an algorithm that, given $k,r \in \N_{\geq 1}$ and a graph $G$ with $\delta(G) \geq k(r
+ 1) - 1$, outputs a member of ${\bf P}_{{\sf v},\Theta_{r}}^{\geq
  k}(G)$ in $O(n^c + m)$ steps, where $c$ is the constant of \autoref{p:bazgan}.
\end{lemma}

\begin{proof}
After applying the algorithm of \autoref{p:bazgan} on $G$ to
obtain in $O(n^c)$-time $k$ graphs $G[V_1], \dots, G[V_k]$, we extract
a $\Theta_r$-subdivision from each of them using~\autoref{c:extract}.
\end{proof}

\subsection{Approximation algorithms}
\label{s:appx}

As $\Theta_r$-subdivisions are exactly the subgraph minimal graphs that contain $\Theta_r$ as
a minor, a graph has a $\Theta_r$-{\sf x}-packing of size $k$ iff it contains
$k$ {\sf x}-disjoint subgraphs, each containing $\theta_r$ as a minor.
Therefore, \autoref{madin} is a direct combinatorial consequence of the following.

\begin{theorem}\label{t:appxep}
  There is a function $\newfun{appx} \colon \N \to \N$ and an algorithm that, with
  input  ${\sf x}\in\{{\sf v},{\sf e}\}$, $r \in \N_{\geq 2}$, $k \in
  \N$, and an $n$-vertex graph $W$, outputs either a
  $\x$-$\Theta_{r}$-packing of $W$ of size~$k$ or an
  $\x$-$\Theta_{r}$-covering of $W$ of size at most~$\funref{appx}(r)\cdot k\cdot \log
  k$. Moreover, this algorithm runs in $O(m^2)$ steps if $\x = {\sf
    e}$ and in $O_r(n^{c} + n \cdot m)$ steps if $\x = {\sf
    v}$, where $c$ is the constant from~\autoref{p:bazgan}.
\end{theorem}

\begin{proof}
  Let $\funref{appx}(r): \N\rightarrow \N$ be a function such that each $\Theta_r$-subdivision output by the
  algorithm of \autoref{l:eppump} has size at most $\funref{appx}(r)\cdot \log k$.
  We consider the following procedure.
  \begin{enumerate}
  \item $G:=W$ ; $P := \emptyset$;
  \item \label{e:loop} Apply the algorithm of \autoref{l:eppump}
    on $(\x, r, k,G)$:
    \begin{description}
    \item[Progress:] if the output is a $\Theta_r$-subdivision $M$, let $G:= G\setminus
      A_\x(M)$ and $P = P \cup \{M\}$;
    \item[Win:] if the output is a $H$-minor model $M$ in $W$ for some
      graph $H$ with $\delta(H)\geq k(r+1)$, apply the algorithm
      of \autoref{l:bigepack} (if $\x={\sf e}$) or the one of
      \autoref{l:vpackex} (if $\x={\sf v}$) to $H$ to obtain a
      member of ${\bf P}_{{\sf x},\Theta_{r}}^{\geq k}(H)$. Using $M$,
      translate this packing into a member of ${\bf P}_{{\sf
          x},\Theta_{r}}^{\geq k}(W)$ and return this new packing;

    \item[Reduce:] otherwise, the output is a graph $G'$. Then let $G := G'$.
    \end{description}
    \item If $|P|=k$ then return $P$ which is a member of ${\bf P}_{{\sf
          x},\Theta_{r}}^{\geq k}(W)$;
    \item If $n(W) = 0$ then return $P$ which is in this case a member of ${\bf C}_{{\sf
          x},\Theta_{r}}^{\leq \funref{appx}(r)\log k}(W)$;
    \item Otherwise, go back to Line~\ref{e:loop}.
  \end{enumerate}
This algorithm clearly returns the desired result. Furthermore, the
loop is executed at most $\left | A_{\sf x}(W)\right |$ times (as an
element of $A_{\sf x}(W)$ is deleted each time we reach the
\emph{Progress} case) and each call to the algorithm
of \autoref{l:eppump} takes $O(m(W))$ steps.
When the algorithm reaches the ``Win'' case (which can happen at most
once), the calls to the algorithm of \autoref{l:bigepack} (if
$\x={\sf e}$) or the one of \autoref{l:vpackex} (if
$\x={\sf v}$), respectively, take
$O(m(H))$ and $O\left ((n(H))^c \right )$ steps.
Therefore, in total, this algorithm terminates in $O(m^2)$ steps
if $\x={\sf e}$ and in $O\left(n^c +n \cdot m\right )$ steps if~$\x={\sf v}$.
\end{proof}

Observe that if the algorithm of \autoref{t:appxep} reaches the
``Win'' case, then the input graph is known to contain an $\x$-$\Theta_r$-packing
of size at least $k$. As a consequence, if we are only interested in the
existence of a packing or covering, the call to the algorithm of
\autoref{l:bigepack} or \autoref{l:vpackex} is not necessary.

\begin{corollary}\label{c:existentialist}
  There is an algorithm that, with input  ${\sf x}\in\{{\sf v},{\sf
    e}\}$, $r \in \N_{\geq 2}$, $k \in \N$, and a graph
  $W$, outputs 0 only if $W$ has an $\x$-$\Theta_{r}$-packing of
  size~$k$ or 1 only if $W$ has an $\x$-$\Theta_{r}$-covering of
  size at most~$\funref{appx}(r)\cdot k\cdot \log k$. Furthermore this algorithm
  runs in $O_r(n \cdot m)$ steps when $\x = {\sf v}$ and in $O_r(m^2)$
  steps when~$\x = {\sf e}$.
\end{corollary}

We now conclude this section with the proof of  \autoref{admndi}.
For the same reason as in the proof of \autoref{madin} at the
beginning of this section, we can deal with $\Theta_r$-subdivisions instead
of dealing with subgraphs that contain $\theta_r$ as a minor.

\begin{proof}[Proof of \autoref{admndi}]
Let us call ${\sf A}$ the algorithm of \autoref{c:existentialist}.
Let $k_0\in \intv{1}{n(W)}$ be an integer such that ${\sf A}(\x, r, k_0, W)
= 1$ and ${\sf A}(\x, r, k_0-1, W) = 0$, and let us show that the
value $\funref{appx}(r)\cdot k_0
\log k_0$ is an $O(\log OPT)$-approximation of $\xpack_{\Theta_r}(W)$
and $\cover_{\Theta_r}(W)$.

First, notice that for every $k > \xpack_{\Theta_r}(W)$, the value
returned by ${\sf A}(\x, r, k,$\\
$W)$ is 1. Symmetrically, for every $k$
such that $\funref{appx}(r)\cdot k\log k < \xcover_{\Theta_r}(W)$, the value of ${\sf A}(\x, r, k,
W)$ is 0.
Therefore, the value $k_0$ is such that:
\begin{align*}
k_0-1 &\leq  \xpack_{\Theta_r}(W)\ \text{and}\\
\xcover_{\Theta_r}(W) &\leq  \funref{appx}(r)\cdot k_0\log k_0.
\end{align*}
As every minimal covering must contain at least one vertex or edge
(depending on whether $\x = {\sf v}$ or $\x = {\sf e}$) of each element (i.e.\ subdivision) of
a maximal packing  $\xpack_{\Theta_r}(W) \leq
\xcover_{\Theta_r}(W)$, we have the following two equations:
\begin{align}
 \xpack_{\Theta_r}(W) \leq  \funref{appx}(r)\cdot k_0\log k_0 &\leq
  (\xpack_{\Theta_r}(W)+1)\log (\xpack_{\Theta_r}(W)+1)\label{e:1}\\
\xcover_{\Theta_r}(W) \leq  \funref{appx}(r)\cdot k_0\log k_0 &\leq  (\xcover_{\Theta_r}(W)+1)\log (\xcover_{\Theta_r}(W)+1)\label{e:2}.
\end{align}
Dividing \eqref{e:1} by $\xpack_{\Theta_r}(W)$ and \eqref{e:2} by $\cover_{\Theta_r}(W)$, we get:
\begin{align*}
1 \leq \frac{k_0\log k_0}{\xpack_{\Theta_r}(W)} &\leq \log
(\xpack_{\Theta_r}(W) + 1) +\frac{\log \xpack_{\Theta_r}(W)}{
  \xpack_{\Theta_r}(W)}\\ &= O(\log (\xpack_{\Theta_r}(W))), \quad \text{and}\\
1 \leq \frac{k_0\log k_0}{\xcover_{\Theta_r}(W)} &\leq \log
(\xcover_{\Theta_r}(W) + 1)+\frac{\log \xcover_{\Theta_r}(W)}{
  \xcover_{\Theta_r}(W)}\\ &= O(\log (\xcover_{\Theta_r}(W))).
\end{align*}

Therefore the value $k_0\log k_0$ is both an $O(\log
OPT)$-approximation of $\xpack_{\Theta_r}(W)$ and
$\cover_{\Theta_r}(W)$.
The value $k_0$ can be found by performing a binary search in the interval
$\intv{1}{n}$, with $O(\log n)$ calls to Algorithm ${\sf A}$. Hence, our
approximation algorithm runs in $O(n\cdot\log(n)\cdot m)$~steps when
$\x = {\sf v}$ and in $O(m^2\cdot\log(n))$~steps when~$\x = {\sf e}$.
\end{proof}

\section{\texorpdfstring{Erdős–Pósa}{Erdös-Posa} property and tree-partition width}
\label{i9o0i4y}
%
%This section contains the proofs of \autoref{t:epbound} and \autoref{t:epdoublepk}.
%
%
%
%\subsection{Proofs of \autoref{t:epbound} and \autoref{t:epdoublepk}}

 Using the machinery introduced in
\autoref{sec:red}, we are able to prove a slightly more general
version of \autoref{t:epbound}.

\begin{theorem}\label{t:epbound2}
 Let $t \in \N$. For every ${\sf x}\in\{{\sf v},{\sf e}\}$, the
  following holds: if ${\cal H}$ is a finite collection of connected graphs and $G$ is
  a graph of tree-partition width at most $t$, then ${\sf x}$-\cover$_{\cal
    H}(G)\leq \alpha \cdot {\sf x}$-\pack$_{\cal H}(G)$, where $\alpha$ is a constant which depends only on $t$ and $\mathcal{H}$.
\end{theorem}

\begin{proof}
Let $k = \xpack_{\cal H}(G)$. This proof is similar to that of
\autoref{t:appxep} (progress or reduce).
We start with $P :=
\emptyset$ and $G_0 := G$ and repeatedly apply \autoref{l:new_reduce}
to $G_i$, which provides the existence of either a smaller graph $G_{i+1}$
with the same parameters $\xcover_{\cal H}$ and $\xpack_{\cal H}$ (reduce), or
that of an ${\cal H}$-subdivision $M$ with at most $\funref{f:newred}(h, t)$
edges, in which case we set $G_{i+1}:=G_i\setminus A_{\x}(M)$ and $P := P \cup
\{M\}$ (progress) and continue. Notice that while the algorithm of
\autoref{l:new_reduce} requires a $t$-partitioned protrusion, it is
enough here to know that such a partitioned protrusion exists as we do
not actually run the algorithm. This is the case because the
considered graphs are subgraph of $G$, which has tree-partition width
at most $t$, so we can consider a $t$-partitioned protrusion
containing the whole graph.
We stop this process when the current graph has at most $\funref{f:newred}(h, t)$
vertices. Let $G_j$ be this graph. In the end, $P$ contains at most $k$ ${\cal H}$-subdivisions.
Therefore, $C := \cupall A_{\x}(P) \cup A_{\x}(G_j)$ is an
$\x$-$\mathcal{H}$-covering of $G$ of size $(k+1) \cdot
\funref{f:newred}(h, t)$, as required.
\end{proof}

\autoref{t:epbound} can be obtained from \autoref{t:epbound2} by
considering packing and covering of the graphs in $\bigcup_{H \in
\mathcal{H}} \ex(H)$, where $\mathcal{H}$ is the family of graphs mentioned
in the statement of \autoref{t:epbound}.

% \begin{corollary}
% For every graph $H$ and every $t \in \N$, the class of $H$-minor-models has
% the (edge and vertex) Erdős–Pósa property for graphs of tree-partition width at
% most~$t$. Furthermore the gap is a linear function.
% \end{corollary}

We define  $\Theta_{r,r'}=\ex(\theta_{r,r'})$. 
 The rest of this section is devoted to the proof of~\autoref{t:epdoublepk}.
Prior to this, we need to
introduce a result of Ding {\em et al.}~\cite{Ding199645}.

Tree-partition width has been studied in~\cite{Seese, Halin1991203,
  Ding199645}.
In particular, the authors of~\cite{Ding199645} characterized the
classes of graphs of bounded tree-partition width in terms of excluded
topological minors. The statement of this result requires additional
definitions.

\paragraph{Walls, fans, paths, and stars.} The $n$-wall is the graph with vertex set $\intv{1}{n}²$ and whose
edge set~is:
\begin{align*}
  &\left \{ \{(i,j), (i, j+1)\},\ 1\leq i,j \leq n \right \}\\
  \cup &\left \{ \{(2i-1,2j+1), (2i, 2j+1)\},\ 1<2i\leq n\ \text{and}\
         1\leq 2j +1\leq n\right \}\\
  \cup &\left \{ \{(2i,2j), (2i+1, 2j)\},\ 1\leq 2i < n\ \text{and}\
         1\leq 2j \leq n \right \}.
\end{align*}
The 7-wall is depicted in~\autoref{fig:dingraph}.
The \emph{$n$-fan} is the graph obtained by adding a dominating vertex
to a path on $n$~vertices. A collection of paths is said to be
\emph{independent} if two paths of the collection never share interior
vertices. The \emph{$n$-star} is the graph obtained by replacing every edge
of $K_{1,n}$ with $n$ independent paths of two edges. The \emph{$n$-path} is
the graph obtained by replacing every edge of an $n$-edge path with $n$
independent paths of two edges. Examples of these graphs are depicted
in~\autoref{fig:dingraph}.
The \emph{wall number} (resp.\ \emph{fan number}, \emph{star number}, and
\emph{path number}) of a graph $G$ is defined as the largest integer
$k$ such that $G$ contains a minor model of a $k$-wall (resp. of a $k$-fan,
of a $k$-star, of a $k$-path), or infinity is no such integer exists.
Let~$\gamma(G)$ denote the maximum of the wall number, fan number, star
number, and path number of a graph~$G$.

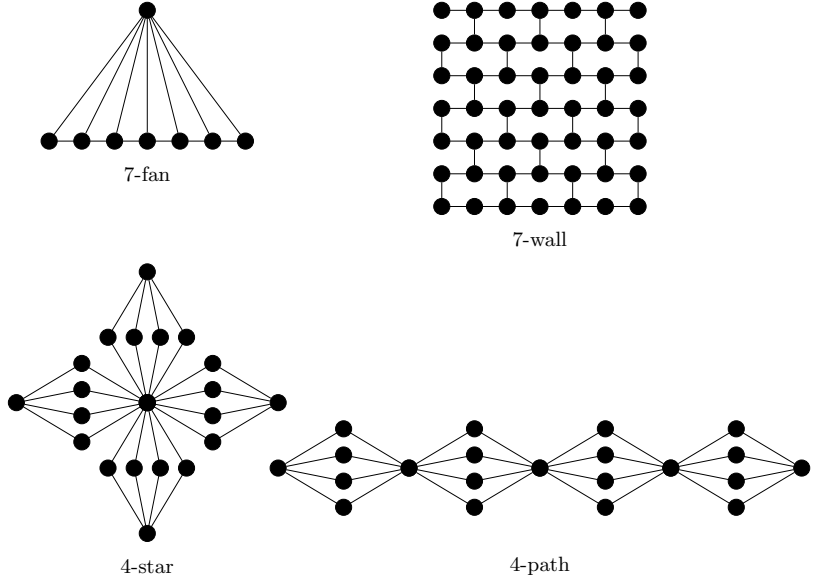
\begin{figure}
  \centering
  \scalebox{.87}{
  \begin{tikzpicture}[every node/.style = black node]
    \begin{scope}[xscale = 0.5]% fan
      \draw (0,0) node (n) {}
      \foreach \i in {-3, ..., 3}{-- (\i, -2) node (n\i) {}};
      \foreach \i in {-2, ..., 3}{\draw (n) -- (n\i);}
      \node[normal] at (0,-2.5) {\small 7-fan};
    \end{scope}
    \begin{scope}[yshift = -6cm] % 4-star
      \draw (0,0) node (n0) {} (2,0) node (n2) {};
      \foreach \i in {-0.6,-0.2, 0.2, 0.6}
      {\draw (n0) -- (1, \i) node {} -- (n2);}
      \begin{scope}[rotate = 90]
        \draw (0,0) node (n0) {} (2,0) node (n2) {};
        \foreach \i in {-0.6,-0.2, 0.2, 0.6}
        {\draw (n0) -- (1, \i) node {} -- (n2);}
      \end{scope}
      \begin{scope}[rotate = 180]
        \draw (0,0) node (n0) {} (2,0) node (n2) {};
        \foreach \i in {-0.6,-0.2, 0.2, 0.6}
        {\draw (n0) -- (1, \i) node {} -- (n2);}
      \end{scope}
      \begin{scope}[rotate = -90]
        \draw (0,0) node (n0) {} (2,0) node (n2) {};
      \foreach \i in {-0.6,-0.2, 0.2, 0.6}
      {\draw (n0) -- (1, \i) node {} -- (n2);}
      \end{scope}
      \node[normal] at (0,-2.5) {\small 4-star};
    \end{scope}
      \begin{scope}[xshift = 4.5cm, yshift = -3cm, scale = 0.5] % 7-wall
        \foreach \i in {0,...,6}{ % draw horizontal lines + vertices
          \draw (0, \i) node (N\i N0) {}
          \foreach \j in {1,...,6}{
            -- (\j, \i) node (N\i N\j) {}
          };
        }
        \foreach \i in {0,2,...,4}{
          \foreach \j in {0,2,...,6}{
            \pgfmathparse{int(\i + 1)}
            \draw (N\i N\j) -- (N\pgfmathresult N\j);
          }}
        \foreach \i in {1,3,...,5}{
          \foreach \j in {1,3,...,5}{
            \pgfmathparse{int(\i + 1)}
            \draw (N\i N\j) -- (N\pgfmathresult N\j);
          }}
        \node[normal] at (3,-1) {\small 7-wall};
      \end{scope}
      \begin{scope}[xshift = 2cm, yshift = -7cm]% 4-path
        \draw (0,0) node (n0) {} (2,0) node (n2) {};
      \foreach \i in {-0.6,-0.2, 0.2, 0.6}
      {\draw (n0) -- (1, \i) node {} -- (n2);}
      \begin{scope}[xshift = 2cm]
        \draw (0,0) node (n0) {} (2,0) node (n2) {};
      \foreach \i in {-0.6,-0.2, 0.2, 0.6}
      {\draw (n0) -- (1, \i) node {} -- (n2);}
      \end{scope}
      \begin{scope}[xshift = 4cm]
        \draw (0,0) node (n0) {} (2,0) node (n2) {};
      \foreach \i in {-0.6,-0.2, 0.2, 0.6}
      {\draw (n0) -- (1, \i) node {} -- (n2);}
      \end{scope}
      \begin{scope}[xshift = 6cm]
        \draw (0,0) node (n0) {} (2,0) node (n2) {};
      \foreach \i in {-0.6,-0.2, 0.2, 0.6}
      {\draw (n0) -- (1, \i) node {} -- (n2);}
      \end{scope}
      \node[normal] at (4,-1.5) {\small 4-path};
    \end{scope}
  \end{tikzpicture}
  }
  \caption{Unavoidable patterns of graphs of large tree-partition width.}
  \label{fig:dingraph}
\end{figure}
We need the following result.
\begin{proposition}[\!\!\!\cite{Ding199645}]\label{p:ding}
There is a function $\newfun{f:tied} \colon \N \to \N$ such that every
graph $G$ satisfies $\tpw(G) \leq \funref{f:tied}(\gamma(G))$.
\end{proposition}

In other words, for every integer $k$, every graph of
large enough tree-partition width contains a minor model of one of the
following graphs: the $k$-wall, the $k$-fan, the $k$-path, or the $k$-star.

 Notice that for every $r,r'
\in \N$, $r'\leq r$, the graph $\theta_{r,r'}$ is a minor
of the following graphs: the $r$-path, the $r$-star, the
$(r+r'+1)$-fan, and the $r$-wall (for $r\geq 6$). Hence, every graph of
large enough tree-partition width contains a 
$\Theta_{r,r'}$-subdivision. This can easily be generalized to
edge-disjoint packings, as follows.

\begin{lemma}
  For every $r,r' \in \N,$ $r'\leq r$, and every
  $k \in \N_{\geq1}$, every graph $G$ satisfying $\gamma(G) \geq
   k(r+r'+2)-1$ contains an ${\sf e}$-${\Theta_{r,r'}}$-packing of
   size~$k$.
\end{lemma}
Using \autoref{p:ding}, we get the following corollary.
\begin{corollary}\label{c:boundtpw}
  For every $r,r' \in \N,$ $r'\leq r$, and every
  $k \in \N_{\geq1}$, every graph $G$ satisfying $\tpw(G) \geq
   \funref{f:tied}(k(r+r'+2)-1)$ contains an ${\sf e}$-${\Theta_{r,r'}}$-packing of
   size~$k$.
\end{corollary}

We are now able to give the proof of~\autoref{t:epdoublepk}.
As in the proof of \autoref{madin} at the beginning of
\autoref{s:appx}, we can consider $\Theta_{r,r'}$-subdivisions instead of
sugraphs containing $\theta_{r,r'}$ as a minor.
\begin{proof}[Proof of~\autoref{t:epdoublepk}]
  According to~\autoref{c:boundtpw}, for every $k \in \N$, there
  is a number $t_k$ such that every graph $G$ with
  $\ecover_{\Theta_{r,r'}}(G) =  k$ satisfies $\tpw(G) \leq
  t_k$. Indeed, such a graph does not contain a packing of $k+1$
  ${\Theta_{r,r'}}$-subdivisions. Then by~\autoref{t:epbound} the value $\ecover_{\Theta_{r,r'}}(G)$
  is bounded above by $\funref{f:newred}(h,t_k) \cdot \epack_{\Theta_{r,r'}}(G)$, and this concludes the~proof.
\end{proof}

\section{Concluding remarks}
\label{sec:concl}

The main algorithmic contribution of this paper is a $\log({\sf
  OPT})$-approximation algorithm for the parameters $\vpack_{\theta_{r}}$,
$\vcover_{\theta_{r}}$, $\epack_{\theta_{r}}$, and
$\ecover_{\theta_{r}}$, for every positive integer~$r$. This improves
the results of~\cite{JoretPSST14hitt} in the case of vertex packings and coverings and
is the first approximation algorithm for the parameters
$\epack_{\theta_{r}}$ and $\ecover_{\theta_{r}}$ for general~$r$.
Our proof uses a reduction technique of independent interest, which is not specific to the graph
$\theta_r$ and can be used for any other (classes of) graphs.

On the combinatorial side, we optimally improved the gap of the
edge-Erdős-Pósa property of minor models of $\theta_r$ for
every~$r$. Also, we were able to show that every class of graphs has
the (edge and vertex) Erdős-Pósa property in graphs of bounded
tree-partition width, with linear gap. An other outcome of this work
is that minor models of $\theta_{r,r'}$ have the
edge-Erdős-Pósa. Recall that prior to this work, the only graphs for
which this was known were~$\theta_r$'s.

As mentioned in~\cite{2013arXiv1311.1108R}, the planarity of a graph
$H$ is a necessary condition for the minor models of $H$ to have the 
edge-Erdős-Pósa property. However, little is known on which planar graphs
have this property and with which gap. This is the first
direction of research that we want to highlight here. Also, the
question of an approximation algorithm can be asked for packing and
covering the minor models of different graphs.
It was proved in~\cite{ChekuriC13larg} that the gap of the vertex-Erdős-Pósa property
of minor models of every planar graph is $O(k
\mathop{\mathrm{polylog}} k)$. It would be interesting to check if these
results can be used to derive a $\mathop{\mathrm{polylog}}({\sf
  OPT})$-approximation for vertex packing and covering minor models of
any planar graph. 

Notice that all  our results are strongly exploiting Lemma~\ref{l:new_reduce}
that holds for every finite collection ${\cal H}$
of connected graphs. Actually, what 
is missing in order to have 
an overall generalization of all of our results, is an 
extension of Proposition~\ref{rk6ter} where 
$\Theta_{r}$ is replaced by any finite collection ${\cal H}$
of connected planar graphs. 
This is an an interesting combinatorial 
problem even for particular instantiations of ${\cal H}$.

\bibliographystyle{abbrv}
%\bibliography{../pumpkil.bib}

\end{document}